\documentclass[final,3p,times]{elsarticle}

\usepackage{amssymb}
\usepackage{amsmath}
\usepackage{amsthm}
\usepackage{bm}

\newtheorem{thm}{Theorem}
\newtheorem{lemma}{Lemma}
\newtheorem{corollary}{Corollary}
\newdefinition{definition}{Definition}
\newtheorem{proposition}{Proposition}
\newproof{decision}{Decision Problem}

\journal{a journal}

\begin{document}

\begin{frontmatter}


\title{Characterization of canonical systems with six types of coins \\
for the change-making problem}

\author[1]{Yuma Suzuki}
\author[2]{Ryuhei Miyashiro\corref{cor1}}
\ead{r-miya@cc.tuat.ac.jp}
\cortext[cor1]{Corresponding author}
\affiliation[1]{organization={Graduate School of Engineering, Tokyo University of Agriculture and Technology},
addressline={2-24-16 Naka-cho},
city={Koganei-shi},
postcode={Tokyo 184-8588},
country={Japan}}
\affiliation[2]{organization={Institute of Engineering, Tokyo University of Agriculture and Technology},
addressline={2-24-16 Naka-cho},
city={Koganei-shi},
postcode={Tokyo 184-8588},
country={Japan}}

\begin{abstract}
This paper analyzes a necessary and sufficient condition for the change-making problem to be solvable with a greedy algorithm.
The change-making problem is to minimize the number of coins used to pay a given value in a specified currency system.
This problem is NP-hard, and therefore the greedy algorithm does not always yield an optimal solution.
Yet for almost all real currency systems, the greedy algorithm outputs an optimal solution.
A currency system for which the greedy algorithm returns an optimal solution for any value of payment is called a canonical system.
Canonical systems with at most five types of coins have been characterized in previous studies.
In this paper, we give characterization of canonical systems with six types of coins, and we propose a partial generalization of characterization of canonical systems.
\end{abstract}




\end{frontmatter}


\section{Introduction} \label{Section1}

The change-making problem is to minimize the number of coins used to pay a given value~$v$ in a currency system (hereinafter, system) $C = (c_1, c_2, \ldots, c_n)$, where $v$ and $c_i \ (i=1,2,\ldots,n)$ are positive integers, $c_i$ is the value of the $i\/$th type of coin in $C$, and $c_1 < c_2 < \cdots < c_n$.
Throughout this paper, we fix $c_1=1$ so that any value $v$ is payable in $C$.

The change-making problem is a special case of the knapsack problem and is known to be NP-hard~\cite{KoZa94}.
Thus, a polynomial-time algorithm for this problem is unlikely to exist unless $\mbox{P} = \mbox{NP}$, whereas several pseudo polynomial-time algorithms based on dynamic programming have been proposed to date; see, for example, \cite{ChHe20b}.

A simple algorithm based on the greedy principle is to repeatedly pay the coin whose value is largest but less than or equal to the rest of the value unpaid.
This greedy algorithm, of course, does not necessarily produce an optimal solution.
For example, to pay the value $v = 6$ in the system $C = (1, 3, 4)$, the greedy algorithm returns three coins ($6 = 4 + 1 + 1$), whereas the optimal solution involves only two coins ($6 = 3 + 3$).
However, for almost all real systems, the greedy algorithm yields an optimal solution for any value of payment.

For a given system, we refer to a value such that the greedy algorithm does not yield an optimal solution as a counterexample to the system.
If a system has no counterexample, we say that the system is canonical.

This paper considers a necessary and sufficient condition for systems to be canonical, that is, characterization of canonical systems.
Characterization of canonical systems was obtained for systems with up to five types of coins in previous studies.
The contribution of the present study is to characterize canonical systems with six types of coins.
In addition, a partial generalization of the characterization of canonical systems is given.

The rest of this paper is organized as follows.
Section~\ref{Section2} formally defines the change-making problem and the decision problem for whether a given system is canonical and introduces related results.
Section~\ref{Section3} derives characterization of canonical systems with six types of coins.
Finally, Section~\ref{Section4} describes a partial generalization of the characterization of canonical systems and presents our conclusions.

\section{Change-making problem and canonical systems} \label{Section2}

\subsection{Definition and characterization of canonical systems} \label{Section2.1}

The change-making problem is to minimize the number of coins used to pay a given value~$v$ in a system $C = (c_1, c_2, \ldots, c_n)$, where $v$ and $c_i \ (i = 1, 2, \ldots, n)$ are positive integers, $c_i$ is the value of the $i\/$th type of coin in $C$, and $1 = c_1 < c_2 < \cdots < c_n$.
The problem can be naturally formulated as the following integer programing problem:
\begin{align}
\text{minimize}   & \quad \sum_{i=1}^n x_i \notag \\
\text{subject~to} & \quad \sum_{i=1}^n c_i x_i = v, \notag \\
                  & \quad x_i \in \mathbb{Z}_{\geq 0} \quad(i = 1, 2, \ldots, n), \notag
\end{align}
where the nonnegative integer variable $x_i \ (i = 1, 2, \ldots, n)$ corresponds to the number of coins whose value is $c_i$ involved in paying the value $v$.

For given $v$ and $C$, we refer to a vector $\bm{x} = (x_1, x_2, \ldots, x_n)$ of a feasible solution of the integer programming problem as a \emph{representation} of~$v$ in $C$.
Similarly, for given $v$ and $C = (c_1, c_2, \ldots, c_n)$, an \emph{optimal representation} is an optimal solution vector for~$v$, and the \emph{greedy representation} is the feasible solution vector for $v$ given by the greedy algorithm in Table~\ref{Tab:GRDALGO}.
Note that for some $v$ and $C$, there may be multiple optimal representations, whereas the greedy representation is unique for any $v$ and $C$.
For instance, the optimal representations for $v = 12$ in $C = (1, 4, 6, 8)$ are $\bm{x} = (0, 1, 0, 1)$ and $(0, 0, 2, 0)$.

\begin{table}
\caption{Greedy algorithm}\label{Tab:GRDALGO}
\centering
\begin{tabular}{l}
\hline
\noalign{\smallskip}
\textbf{Input} $v$ and $C = (c_1, c_2, \ldots, c_n)$. \\
\textbf{Set} $\bm{x}=(x_1, x_2, \ldots, x_n)$ := $(0, 0, \ldots, 0)$. \\
\textbf{For} $i := n$ \textbf{downto} $1$ \textbf{do}: \\
\qquad\textbf{While} $c_i \leq v$ \textbf{do}: \\
\qquad\qquad $v$ := $v - c_i$ and $x_i$ := $x_i + 1$. \\
\textbf{Output} $\bm{x}=(x_1, x_2, \ldots, x_n)$. \\
\noalign{\smallskip}
\hline
\end{tabular}
\end{table}

Denote the total number of coins used in an optimal representation for $v$ in~$C$ by $\text{opt}_{C}(v)$, that is, $\text{opt}_{C}(v) = \sum_{i=1}^n x_i^{*}$ where $\bm{x}^{*}=(x_1^{*}, x_2^{*}, \ldots, x_n^{*})$ is an optimal representation for $v$ in~$C$.
Similarly, denote the total number of coins used in the greedy representation for~$v$ in~$C$ by~$\text{grd}_{C}(v)$.

We call the value $w \in \mathbb{Z}_{>0}$ a \emph{counterexample} to $C$ if $\text{opt}_{C}(w)< \text{grd}_{C}(w)$.
A system $C$ is called \emph{noncanonical} if there exists a counterexample to $C$; otherwise, $C$ is said to be \emph{canonical}.\footnote{Some terms other than canonical are used in the literature, such as standard~\cite{HuLe76}, greedy~\cite{CoCoSt08}, and orderly~\cite{AdAd10}.}
In other words, if $C$ is canonical, $\text{opt}_{C}(v) = \text{grd}_{C}(v)$ holds for any $v \in \mathbb{Z}_{>0}$. 

The change-making problem is NP-hard in general, but the greedy algorithm produces optimal solutions for almost all practical systems.
Accordingly, the following decision problem has received research attention.

\begin{decision}~\\
Instance: A system $C = (1, c_2, \ldots, c_n)$. \\
Task: Decide whether $C$ is canonical. 
\end{decision}

\noindent
For this decision problem, Chang and Gill~\cite{ChGi70} proposed an $O(c_n^3 n)$ algorithm, and Kozen and Zaks~\cite{KoZa94} proposed an $O(c_n n)$ algorithm.
These methods are polynomial with respect to~$c_n$ but not polynomial with respect to the input size of an instance of the change-making problem.
Later, Pearson~\cite{Pe05} proposed an $O(n^3)$ algorithm by bounding the number of candidates of the minimum counterexample by~$O(n^2)$.

Meanwhile, previous studies have characterized canonical systems with up to five types of coins.
Systems with one or two types of coins are obviously canonical, and the characterization of (non)canonical systems with three types of coins was given by Kozen and Zaks~\cite{KoZa94}.

\begin{thm}[Kozen and Zaks~\cite{KoZa94}] \label{Thm:3coin}
A system $C = (1, c_2, c_3)$ is noncanonical if and only if $0 < r < c_2 - q$ where $c_3 = q c_2 + r$ for $0 \leq r < c_2$.
\end{thm}

\begin{corollary}\label{Cor:3coin}
A system $C = (1, c_2, c_3)$ is canonical if and only if $r = 0$ or $c_2 - q \leq r$ where $c_3 = q c_2 + r$ for $0 \leq r < c_2$.
\end{corollary}
\noindent
The characterization of canonical systems with four types of coins and that with five types of coins were given by Adamaszek and Adamaszek~\cite{AdAd10} and Cai~\cite{Cai09}.
\begin{thm}[Adamaszek and Adamaszek~\cite{AdAd10}, Cai~\cite{Cai09}] \label{Thm:4coin}
A system $C = (1, c_2, c_3, c_4)$ is canonical if and only if the subsystem $(1, c_2, c_3)$ is canonical and $\text{grd}_{C}(mc_3) \leq m$ for $m = \lceil c_4 / c_3 \rceil$.
\end{thm}

\begin{thm}[Adamaszek and Adamaszek~\cite{AdAd10}, Cai~\cite{Cai09}] \label{Thm:5coin}
A system $C = (1, c_2, c_3, c_4, c_5)$ is canonical if and only if (a) or (b) holds:
\begin{itemize}
\item[(a)] the subsystem $(1, c_2, c_3, c_4)$ is canonical and $\text{grd}_{C}(mc_4) \leq m$ for $m = \lceil c_5 / c_4 \rceil;$
\item[(b)] $C = (1, 2, c_3, c_3 + 1, 2c_3)$ and $c_3 > 3$.
\end{itemize}
\end{thm}
\noindent
Note that, in the case of Theorem~\ref{Thm:5coin}(b), the subsystem with the leading four types of coins of $C$, namely, $C'=(1, 2, c_3, c_3 + 1)$ for $c_3 > 3$, is noncanonical because $\text{opt}_{C'}(2c_3) = 2 < \text{grd}_{C'}(2c_3)$.

As above, characterization of canonical systems is known for at most five types of coins to date.
We propose the characterization of canonical systems with six types of coins as follows.

\begin{proposition} \label{Pro:6coins}
A system $C = (1, c_2, c_3, c_4, c_5, c_6)$ is canonical if and only if (a) or (b) holds.
\begin{itemize}
\item[(a)] The subsystem $(1, c_2, c_3, c_4, c_5)$ is canonical and $\text{grd}_{C}(mc_5) \leq m$ holds for $m = \lceil c_6 / c_5 \rceil$.
\item[(b)] The subsystem $(1, c_2, c_3, c_4, c_5)$ is noncanonical and $C$ satisfies (i), (ii), or (iii) for $\ell = \lceil c_5 / c_3 \rceil$.
In addition, $\text{grd}_C(\ell c_3) = \ell c_3 - c_5 + 1 - \lfloor (\ell c_3 - c_5) / c_2 \rfloor (c_2 - 1)$.
\begin{itemize}
\item[(i)]   $C = (1, 2, 3, c_4, c_4 + 1, 2 c_4)$ and $c_4 > 4;$
\item[(ii)]  $C = (1, c_2, 2 c_2 - 1, c_4, c_2 + c_4 - 1, 2 c_4 - 1)$, $c_4 \geq 3 c_2-1$, and $\text{grd}_C(\ell c_3) \leq \ell;$
\item[(iii)] $C = (1, c_2, 2 c_2, c_4, c_2 + c_4, 2 c_4)$, $c_4 \geq 3 c_2 - 1$, $c_4 \neq 3 c_2$, and $\text{grd}_C(\ell c_3) \leq \ell$.
\end{itemize}
\end{itemize}
\end{proposition}

\noindent
In Section~\ref{Section3}, we prove that Proposition~\ref{Pro:6coins} is true.

\subsection{Related results on characterization} \label{Section2.2}

More results related to the characterization of canonical systems have been reported other than those already mentioned in Section~\ref{Section2.1}.
Of these results, this subsection describes several definitions and theorems that we need in~Section~\ref{Section3}.

Firstly, we introduce a classic but strong theorem, which is called the ``one-point theorem'' in some papers.

\begin{thm}[Magazine, Nemhauser, and Trotter Jr.~\cite{MaNeTr75}, Hu and Lenard~\cite{HuLe76}, Cowen, Cowen, and Steinberg~\cite{CoCoSt08}] \label{Thm:OPT}
Let $C$ be a system $(1, c_2, \ldots, c_n)$ for $n \geq 2$.
Suppose that the subsystem $C' = (1, c_2, \ldots, c_{n-1})$ of $C$ is canonical.
The following statements are equivalent:
\begin{itemize}
\item[(a)] $C$ is canonical;
\item[(b)] $\text{grd}_{C}(mc_{n-1}) \leq m$ for $m = \lceil c_{n} / c_{n-1} \rceil;$
\item[(c)] $\text{opt}_{C}(mc_{n-1}) = \text{grd}_{C}(mc_{n-1})$ for $m = \lceil c_{n} / c_{n-1} \rceil$.
\end{itemize}
\end{thm}

\noindent
This theorem gives a necessary and sufficient condition for a system $C = (1, c_2, \ldots, c_n)$ to be canonical when the subsystem $C' = (1, c_2, \ldots, c_{n-1})$ is canonical.
Note that the part (a) of Theorem~\ref{Thm:5coin} arises directly by induction from Theorem~\ref{Thm:OPT}.

The following theorem gives upper and lower bounds for the minimum counterexample to a noncanonical system with at least three types of coins; a system with one or two types of coins is always canonical.

\begin{thm}[Kozen and Zaks~\cite{KoZa94}] \label{Thm:w_range}
Assume that a system $C = (1,c_2,\ldots,c_n)$ is noncanonical.
Let $w$ be the minimum counterexample to $C$.
Then $c_3 + 1 < w < c_{n-1} + c_n$.
\end{thm}

Two helpful concepts, a {\tt +/-} class and tight, used in the next section are defined below; the former was developed by Adamaszek and Adamaszek~\cite{AdAd10} and the latter was introduced by Cai~\cite{Cai09}.

\begin{definition}[Adamaszek and Adamaszek~\cite{AdAd10}] \label{Def:PMclass}
For a system $C = (1, c_2, \ldots, c_n)$, a \emph{{\tt +/-} class} of $C$ is a string of length $n$ such that its $i\/$th symbol $(i = 1, 2, \ldots, n)$ is {\tt +} if the subsystem $(1, c_2, \ldots, c_i)$ is canonical, otherwise {\tt -}.
\end{definition}

\begin{definition}[Cai~\cite{Cai09}] \label{Def:tight}
A system $(1, c_2, \ldots, c_n)$ is said to be \emph{tight} if there is no counterexample smaller than $c_n$.
\end{definition}

\begin{corollary} \label{Cor:tight_canonical}
A canonical system is tight. A system that is not tight is noncanonical.
\end{corollary}

\begin{thm}[Cai~\cite{Cai09}] \label{Thm:tight}
Let $n \geq 5$.
Assume that three systems $C_3 = (1, c_2, c_3)$, $C_{n-1} = (1, c_2, \ldots, c_{n-1})$, and $C_n = (1, c_2, \ldots, c_n)$ are tight, $C_3$ is canonical, and $C_{n-1}$ is noncanonical.
Then, if $C_n$ is noncanonical, there exist $i$ and $j$ such that $1 < i \leq j \leq n-1$, $c_i + c_j > c_n$, and $c_i + c_j$ is a counterexample to~$C_n$.
\end{thm}

In addition to the notations $\text{grd}_C (v)$ and $\text{opt}_C (v)$ introduced in Section~\ref{Section2.1}, we define two more.
Let $\text{grd}_{C}^{c_i}(v)$ be the number of coins whose value is $c_i$ used in the greedy representation for $v$ in~$C$, and let $\overline{\text{opt}}\,{}_{C}^{c_i}(v)$ be that used in the lexicographically smallest optimal representation, which is defined below, for $v$ in $C$.

A representation $\bm{x} =(x_1, x_2, \ldots, x_n)$ for $v$ in $C$ is said to be \emph{lexicographically smaller} than a representation $\bm{x}' =(x_1', x_2', \ldots, x_n')$ for $v$ in $C$ if there exists $k \in \{1, 2, \ldots, n-1\}$ such that $x_k < x_k'$ and $x_i = x_i'$ for all $1 \leq i < k$.
As mentioned earlier, there are multiple optimal representations for some $v$ and~$C$.
The \emph{lexicographically smallest} optimal representation for $v$ in~$C$ is the lexicographically smallest representation among optimal representations for~$v$ in~$C$.
For example, for $v = 12$ in $C = (1, 4, 6, 8)$, the optimal representations are $(0, 1, 0, 1)$ and $(0, 0, 2, 0)$, and the latter is the lexicographically smallest optimal representation.

\begin{thm}[Pearson~\cite{Pe05}]\label{Thm:poly}
Assume a system $C = (1, c_2, \ldots, c_n)$ is noncanonical and let $w$ be the minimum counterexample to~$C$.
Then $\overline{\text{opt}}\,{}_{C}^{c_n}(w) = 0$.
Let the lexicographically smallest optimal representation for $w$ in $C$ be
\[
(\overbrace{0, 0, \ldots, 0}^{i-1~0\text{'s}}, x_i, x_{i+1}, \ldots, x_j, \overbrace{0, 0, \ldots, 0}^{n-j~0\text{'s}})
\]
where $1 \leq i \leq j < n$, $x_i > 0$, and $x_j > 0$.
Then the greedy representation for $c_{j+1} - 1$ in $C$ is
\[
(y_1, y_2, \ldots, y_{i-1}, x_i-1, x_{i+1}, \ldots, x_j, \overbrace{0, 0, \ldots, 0}^{n-j~0\text{'s}})
\]
where $y_1, y_2, \ldots, y_{i-1} \in \mathbb{Z}_{\geq 0}$.
\end{thm}

\noindent
Theorem~\ref{Thm:poly} plays a central role in an $O(n^3)$ algorithm for deciding whether a given system is canonical~\cite{Pe05}, but we omit the details in this paper.

The following theorem shows that for a canonical system, the subsystem with the leading three types of coins of the system is also canonical.

\begin{thm}[Adamaszek and Adamaszek~\cite{AdAd10}, Cai~\cite{Cai09}] \label{Thm:pre3}
For $n \geq 3$, if a system $C = (1, c_2, \ldots, c_n)$ is canonical, the subsystem $(1, c_2, c_3)$ is canonical.
\end{thm}

Besides the decision problem for canonical systems, there have been various studies of the change-making problem, and we close this section by simply listing them below.
The change-making problem has some practical applications, such as network design~\cite{GeKl77}, cutting-stock, and capital allocation~\cite{NeUl69}.
From a theoretical perspective, Magazine, Nemhauser, and Trotter Jr.~\cite{MaNeTr75} analyzed conditions for the knapsack problem to be solvable with a greedy method,
Tien and Hu~\cite{TiHu77} studied the gap between greedy and optimal solutions of the change-making problem,
Adamaszek and Adamaszek~\cite{AdAd10} revealed some relationships between canonical systems and their subsystems,
Goebbels, Gurski, Rethmann, and Yilmaz~\cite{GoGuReYi17} considered approximation algorithms and fixed-parameter tractability for the change-making problem,
and Chan and He~\cite{ChHe20b} recently proposed faster dynamic programming-based algorithms for the change-making and related problems.

\section{Characterization of canonical systems with six types of coins} \label{Section3}

This section is devoted to proving Proposition~\ref{Pro:6coins}, which describes characterization of canonical systems with six types of coins at the end of Section~\ref{Section2.1}.
Part (a) of Proposition~\ref{Pro:6coins} immediately follows from Theorem~\ref{Thm:OPT}, and we concentrate on the proof of Proposition~\ref{Pro:6coins}(b), when the subsystem with the leading five types of coins is noncanonical.

Firstly, we give four easy lemmas that are referred to frequently in this section.

\begin{lemma} \label{Lem:prefix}
Assume that a system $C = (1,c_2,\ldots,c_n)$ is canonical.
Then, for each $1 \leq i \leq n$, the subsystem $(1,c_2, \ldots, c_i)$ is tight.
\end{lemma}
\begin{proof}
Suppose that there exists an index $i \ (1 \leq i < n)$ such that the subsystem $(1,c_2,\ldots,c_i)$ is not tight.
Then, there is a counterexample $w_i < c_i$ for the subsystem $(1,c_2,\ldots,c_i)$.
The value $w_i$ is also a counterexample to $C$ because $w_i < c_{i+1} < \cdots < c_n$.
This contradicts the assumption.

When $i=n$, the system $(1,c_2, \ldots, c_i)$ is $C$, and thus it is canonical and tight.
\end{proof}
\noindent
We note that Cai~\cite{Cai09} gave the contraposition of Lemma~\ref{Lem:prefix} without proof.

\begin{lemma} \label{Lem:exCor1}
Suppose that a system $C = (1, c_2, \ldots, c_n)$ is canonical and the subsystem $C' = (1, c_2, \ldots, c_{n-1})$ is noncanonical.
Let $w'$ be the minimum counterexample to $C'$.
Then $c_n \leq w'$ holds.
\end{lemma}
\begin{proof} 
If $w' < c_n$, then $w'$ is also a counterexample to~$C$, which contradicts the assumption that $C$ is canonical.
\end{proof}

\begin{lemma}\label{Lem:PMclass}
If a system $C = (1, c_2, c_3, c_4, c_5, c_6)$ is canonical and the subsystem $(1, c_2, c_3, c_4, c_5)$ is noncanonical, the {\tt +/-} class of $C$ is either {\tt ++++-+} or {\tt +++--+}.
\end{lemma}
\begin{proof}
The first and second symbols of the {\tt +/-} class of $C$ are~{\tt +} because systems with one or two types of coins are canonical.
The third one is also~{\tt +} because the subsystem with the leading three types of coins of a canonical system is canonical (Theorem~\ref{Thm:pre3}).
The fifth and sixth symbols are {\tt -} and~{\tt +}, respectively, from the assumption.
\end{proof}

\begin{lemma} \label{Lem:coin1}
Let $C$ and $w$ be a noncanonical system and its minimum counterexample, respectively.
Then $\overline{\text{opt}}\,{}_{C}^{c_1}(w) = 0$.
\end{lemma}
\begin{proof}
Assume that $\overline{\text{opt}}\,{}_{C}^{c_1}(w)$, namely, $\overline{\text{opt}}\,{}_{C}^{1}(w)$, is greater than zero.
Let the lexicographically smallest optimal representation for $w$ in $C$ be $(x_1, x_2, \ldots, x_j, 0, 0, \ldots, 0)$ where $x_1 > 0$ and $x_j > 0$.
Then, from Theorem~\ref{Thm:poly}, the greedy representation for $c_{j+1} - 1$ in $C$ is $(x_1 -1, x_2, \ldots, x_j, 0, 0, \ldots, 0)$.
Thus, $w = c_{j+1}$ holds.
However, $\text{opt}_{C}(c_{j+1}) = \text{grd}_{C}(c_{j+1}) = 1$, which contradicts $w$ being a counterexample to~$C$.
\end{proof}

The proof of Proposition~\ref{Pro:6coins} proceeds as follows.
Lemmas~\ref{Lem:c6_kouho} and~\ref{Lem:2c5-c4} show that $c_6$ is equal to $2c_5 - c_2$ or $2c_5 - c_3$ if $C = (1, c_2, c_3, c_4, c_5, c_6)$ is canonical and the subsystem $C' = (1, c_2, c_3, c_4, c_5)$ is noncanonical.
Lemma~\ref{Lem:2c5-c2} analyzes the case of $c_6 = 2 c_5 - c_2$, and Lemmas~\ref{Lem:2c5-c3_1}, \ref{Lem:2c5-c3_2}, and~\ref{Lem:2c5-c3_3} handle that of $c_6 = 2 c_5 - c_3$.
Then, based on these lemmas, Theorem~\ref{Thm:need} states a necessary condition that $C$ is canonical and $C'$ is noncanonical.
Theorem~\ref{Thm:sufficiency}, supported by Lemmas~\ref{Lem:g_and_m-1}, \ref{Lem:g_and_m-2}, and \ref{Lem:g_and_m-3}, shows the converse of Theorem~\ref{Thm:need}.
Finally, Theorem~\ref{Crl:iff2} concludes that Proposition~\ref{Pro:6coins} is true.

\begin{lemma} \label{Lem:c6_kouho}
Suppose that a system $C = (1, c_2, c_3, c_4, c_5, c_6)$ is canonical and the subsystem $C' = (1, c_2, c_3, c_4, c_5)$ is noncanonical.
Then $c_6$ is equal to $2c_5 - c_2$, $2c_5 - c_3$, or~$2c_5 - c_4$.
\end{lemma}
\begin{proof}
To prove this lemma, consider paying $2 c_5$ in~$C$.
Let $w'$ be the minimum counterexample to $C'$.
From Lemma~\ref{Lem:exCor1} and Theorem~\ref{Thm:w_range}, we have $c_6 \leq w'$ and $w' < c_4 + c_5 < 2 c_5 < 2 c_6$, respectively.
Thus, we have $c_6 < 2 c_5 < 2 c_6$, which leads to $\text{grd}_{C}^{c_6}(2 c_5) = 1$ and $\text{grd}_{C}(2 c_5) > 1$.
As $C$ is canonical, $\text{opt}_{C}(2 c_5) = \text{grd}_{C}(2 c_5) = 2$.
Since $\text{grd}_{C}^{c_6}(2 c_5) = 1$ holds, $2 c_5 - c_6$ is equal to $1, c_2, c_3, c_4$, or~$c_5$.

Since $c_5 < c_6$ holds, $2 c_5 - c_6 \neq c_5$.
From $c_6 \leq w' < c_4 + c_5 \leq 2 c_5 - 1$, we have $2 c_5 - 1 \neq c_6$, which completes the~proof.
\end{proof}

Lemma~\ref{Lem:c6_kouho} claimed that $c_6$ is equal to $2c_5 - c_2$, $2c_5 - c_3$, or~$2c_5 - c_4$.
However, Lemma~\ref{Lem:2c5-c4} reveals that, in fact, $c_6 \neq 2c_5 - c_4$.

\begin{lemma} \label{Lem:2c5-c4}
If $C = (1,c_2,c_3,c_4,c_5,c_6)$ is canonical and the subsystem $C' = (1,c_2,c_3,c_4,c_5)$ is noncanonical, $c_6 \neq 2c_5-c_4$.
\end{lemma}
\begin{proof}
Assume that $c_6 = 2c_5-c_4$.
Let $w'$ be the minimum counterexample to $C'$.
From Lemma~\ref{Lem:exCor1} and Theorem~\ref{Thm:w_range}, $c_6 \leq w' < c_4+c_5$.
Thus, $\text{grd}_{C}^{c_6}(c_4 + c_5) = 1$ and $\text{grd}_{C}(c_4 + c_5) > 1$.
Since $C$ is canonical, $\text{opt}_{C}(c_4 + c_5) = \text{grd}_{C}(c_4 + c_5) = 2$.
Therefore $c_4+c_5-c_6 = 2c_4-c_5$ is equal to 1, $c_2$, or $c_3$.
Thus, we have
\begin{itemize}
\item[\textbullet] $2c_4-c_5=1 \quad\Leftrightarrow\quad C = (1,c_2,c_3,c_4,2c_4-1,3c_4-2)$,
\item[\textbullet] $2c_4-c_5=c_2 \quad\Leftrightarrow\quad C = (1,c_2,c_3,c_4,2c_4-c_2,3c_4-2c_2)$,
\item[\textbullet] $2c_4-c_5=c_3 \quad\Leftrightarrow\quad C = (1,c_2,c_3,c_4,2c_4-c_3,3c_4-2c_3)$.
\end{itemize}

From Lemma~\ref{Lem:PMclass}, the {\tt +/-} class of $C$ is {\tt ++++-+} or {\tt +++--+}.
Assume that the {\tt +/-} class of $C$ is {\tt ++++-+}.
Then, $C''=(1,c_2,c_3,c_4)$ is canonical.
Applying Theorem~\ref{Thm:5coin}(a) to $C''$, we have that $C'$ is canonical for any of the above three cases, which contradicts the assumption $c_6 = 2c_5-c_4$.

Assume that the {\tt +/-} class of $C$ is {\tt +++--+}.
From Lemma~\ref{Lem:prefix}, both $C'$ and $C'' = (1,c_2,c_3,c_4)$ are tight.
Let $w''$ be the minimum counterexample to $C''$.
If $w'' < c_5$, $w''$ is also a counterexample to $C'$, which contradicts $C'$ being tight.
Therefore $c_5 \leq w''$ holds.
In addition, we have $w'' < c_3+c_4 \leq 2c_4-1$ from Theorem~\ref{Thm:w_range}.
Thus, we have $c_5 \leq w'' < c_3+c_4 \leq 2c_4-1$ and hence $2c_4-c_5 \neq 1$.
From Theorem~\ref{Thm:tight}, there exist $i$ and $j$ such that $1 < i \leq j \leq 4$ and $c_i + c_j$ is a counterexample to~$C'$.
The value $2c_4$ is not a counterexample to $C'$ because now we have that $2c_4$ is equal to $c_2 + c_5$ or $c_3 + c_5$, and thus $\text{grd}_{C'}(2c_4)=2 = \text{opt}_{C'}(2c_4)$. 
Therefore $w' \leq c_3 + c_4$ holds.
With Lemma~\ref{Lem:exCor1}, we have $c_6 \leq w' \leq c_3 + c_4 < 2c_4$.
Hence, $\text{grd}_{C}^{c_6}(2c_4) = 1$ and $\text{grd}_{C}(2c_4) > 1$ hold.
\begin{itemize}
\item[\textbullet]
The case of $2c_4 - c_5 = c_2 \ \Leftrightarrow\ C = (1,c_2,c_3,c_4,2c_4-c_2,3c_4-2c_2)$. \\
Since $C$ is canonical, $2c_4$ is not a counterexample to $C$.
Thus, one of $c_1, c_2, \ldots, c_5$ is equal to $2c_4-c_6 = 2c_2-c_4$.
For $2 \leq i \leq 5$, $c_i = 2c_2-c_4$ leads to $2c_2 = c_i+c_4$, which is a contradiction.
For $i=1$, we have $2c_2 = c_4+1$, which yields $c_3 \leq 2c_2-2$ combined with $c_3 \leq c_4 -1$.
Then, from Theorem~\ref{Thm:3coin}, the system $(1,c_2,c_3)$ is noncanonical, which contradicts the {\tt +/-} class of $C$ being {\tt +++--+}.
\item[\textbullet]
The case of $2c_4-c_5=c_3 \ \Leftrightarrow\  C = (1,c_2,c_3,c_4,2c_4-c_3,3c_4-2c_3)$. \\
Since $C$ is canonical, $2c_4$ is not a counterexample to $C$.
Thus one of $c_1, c_2, \ldots, c_5$ is equal to $2c_4-c_6 = 2c_3-c_4$.
Let $c_i$ be equal to $2c_3-c_4$.
Then we have $C = (1,c_2,c_3,2c_3-c_i,3c_3-2c_i,4c_3-3c_i)$.
Consider paying $2c_3$ in $C'' = (1,c_2,c_3,2c_3-c_i)$.
Clearly $\text{grd}_{C''}(2c_3) = \text{opt}_{C''}(2c_3) =2$, and $2c_3$ is not a counterexample to $C''$.
From Theorem~\ref{Thm:OPT}, $C''$ is canonical, which contradicts the {\tt +/-} class of $C$ being {\tt +++--+}.
\end{itemize}
\end{proof}

At this point, we conclude that if a system $C=(1,c_2,c_3,c_4,c_5,c_6)$ is canonical and the subsystem $C'=(1,c_2,c_3,c_4,c_5)$ is noncanonical, then $c_6$ is equal to $2c_5-c_2$ or $2c_5-c_3$.
Lemmas~\ref{Lem:2c5-c2} and~\ref{Lem:2c5-c3_1} analyze the cases of $c_6=2c_5-c_2$ and $c_6=2c_5-c_3$, respectively.

\begin{lemma} \label{Lem:2c5-c2}
If $C = (1,c_2,c_3,c_4,c_5,c_6)=(1,c_2,c_3,c_4,c_5,2c_5-c_2)$ is canonical and the subsystem $C'=(1,c_2,c_3,c_4,c_5)$ is noncanonical, $C = (1,2,3,c_4,c_4+1,2c_4)$ and $c_4 > 4$.
\end{lemma}
\begin{proof}
Let $w'$ be the minimum counterexample to $C'$.
From Lemma~\ref{Lem:exCor1} and Theorem~\ref{Thm:w_range}, $c_6 \leq w' < c_4+c_5$.
Thus $\text{grd}_{C}^{c_6}(c_4+c_5)=1$ and $\text{grd}_{C}(c_4+c_5) > 1$.
Since $C$ is canonical, $\text{opt}_{C}(c_4+c_5) = \text{grd}_{C}(c_4+c_5) = 2$.
Thus $c_4+c_5-c_6 = c_2+c_4-c_5$ is equal to $1$, $c_2$, or $c_3$.
According to $c_4 < c_5$, $c_2+c_4-c_5 = 1$, which leads to $C = (1,c_2,c_3,c_4,c_2+c_4-1,c_2+2c_4-2)$.

Assume $c_2 > 2$.
From Lemma~\ref{Lem:PMclass}, the {\tt +/-} class of $C$ is {\tt ++++-+} or {\tt +++--+}.

\begin{itemize}
\item[\textbullet] When the {\tt +/-} class of $C$ is {\tt ++++-+}. \\
From Theorem~\ref{Thm:OPT}, if $C'$ is noncanonical, $\text{grd}_{C'}(2c_4) > \text{opt}_{C'}(2c_4)$ holds.
Thus $2c_4$ is a counterexample to~$C'$, which leads to $w' \leq 2c_4$.
\item[\textbullet] When the {\tt +/-} class of $C$ is {\tt +++--+}. \\
Since $C$ is canonical, the systems $(1,c_2,c_3)$, $(1,c_2,c_3,c_4)$, and $C' = (1,c_2,c_3,c_4,c_5)$ are tight from Lemma~\ref{Lem:prefix}.
From the assumption regarding the {\tt +/-} class, $(1,c_2,c_3)$, $(1,c_2,c_3,c_4)$, and $C'$ are canonical, noncanonical, and noncanonical, respectively.
Hence there exists a counterexample $c_i + c_j$ to $C'$ such that $1 < i \leq j \leq 4$ from Theorem~\ref{Thm:tight}.
Thus we have $w' \leq 2c_4$.
\end{itemize}
As above, $w' \leq 2c_4 < 2c_4+c_2 -2 = c_6$ and this contradicts $c_6 \leq w'$.
Hence, we have that $c_2=2$ and $C = (1,2,c_3,c_4,c_4+1,2c_4)$.

Assume $c_3 > 3$.
Consider paying $c_3+c_4$ in $C'$.
Clearly $\text{opt}_{C'}(c_3+c_4) \leq 2$.
In addition, $\text{grd}_{C'}^{c_5}(c_3+c_4) = 1$ because $c_3+c_4 - c_5 = c_3-1$.
Since $c_3+c_4-c_5 = c_3-1 > 2$ and $c_2=2$, we have $\text{grd}_{C'}(c_3+c_4) \geq 3 > \text{opt}_{C'}(c_3+c_4)$ and $c_3 +c_4$ is a counterexample to $C'$.
Therefore $w' \leq c_3+c_4 < 2c_4 = c_6$; however, we already have $c_6 \leq w'$.
Hence, $c_3=3$ and $C = (1,2,3,c_4,c_4+1,2c_4)$.

Assume $c_4 = 4$.
Then $C$ becomes $(1,2,3,4,5,8)$.
Applying Theorem~\ref{Thm:OPT} repeatedly, we have that $C'=(1,2,3,4,5)$ is canonical, which contradicts the assumption, and therefore $c_4 > 4$.
\end{proof}

\begin{lemma} \label{Lem:2c5-c3_1}
If $C = (1,c_2,c_3,c_4,c_5,c_6) = (1,c_2,c_3,c_4,c_5,2c_5-c_3)$ is canonical and the subsystem $C'=(1,c_2,c_3,c_4,c_5)$ is noncanonical, $C$ is $(1,c_2,2c_2-1,c_4,c_2+c_4-1,2c_4-1)$ or $(1,c_2,2c_2,c_4,c_2+c_4,2c_4)$.
\end{lemma}
\begin{proof}
Let $w'$ be the minimum counterexample to $C'$.
From Lemma~\ref{Lem:exCor1} and Theorem~\ref{Thm:w_range}, $c_6 \leq w' < c_4+c_5$ holds.
Since $c_6 < c_4+c_5$, we have $\text{opt}_C(c_4+c_5)=2$.
As $C$ is canonical, $\text{grd}_C(c_4+c_5)=2$ and $\text{grd}_C^{c_6}(c_4+c_5)=1$.
Since $c_4+c_5-c_6 = c_3+c_4 - c_5 < c_3$, we have that $c_3+c_4 - c_5$ is equal to $1$ or $c_2$.

If $c_3+c_4-c_5=1$, then $C=(1,c_2,c_3,c_4,c_3+c_4-1,c_3+2c_4-2)$.
From Lemma~\ref{Lem:PMclass}, the {\tt +/-} class of $C$ is {\tt ++++-+} or {\tt +++--+}.
Firstly, assume that the {\tt +/-} class of $C$ is {\tt ++++-+}.
From Theorem~\ref{Thm:OPT}, if $C'$ is noncanonical, $\text{grd}_{C'}(2c_4) > \text{opt}_{C'}(2c_4)$ holds.
Thus $2c_4$ is a counterexample to~$C'$.
Hence, $w' \leq 2c_4 < 2c_4 + c_3 -2 =c_6$, which is a contradiction.
Next, assume that the {\tt +/-} class of $C$ is {\tt +++--+}.
From Theorem~\ref{Thm:tight}, there exist $i$ and $j$ such that $1 < i \leq j \leq 4$ and $c_i + c_j$ is a counterexample to $C'$.
Thus we have $w' \leq c_i + c_j \leq 2c_4 < c_3+2c_4-2 = c_6$, which is a contradiction.
Therefore $c_3+c_4-c_5 \neq 1$.

If $c_3+c_4-c_5=c_2$, then $C = (1,c_2,c_3,c_4,c_3+c_4-c_2,c_3+2c_4-2c_2)$.
From Theorem~\ref{Thm:3coin}, if $c_3 < 2c_2-1$ then $(1,c_2,c_3)$ is noncanonical, which contradicts the fact that the subsystem $(1,c_2,c_3)$ of a canonical system is canonical.
Thus we have $2c_2-1 \leq c_3$.
From Lemma~\ref{Lem:PMclass}, the {\tt +/-} class of $C$ is {\tt ++++-+} or {\tt +++--+}.
Firstly, assume that the {\tt +/-} class of $C$ is {\tt ++++-+}.
From Theorem~\ref{Thm:OPT}, if $C'$ is noncanonical, $\text{grd}_{C'}(2c_4) > \text{opt}_{C'}(2c_4)$.
Thus $2c_4$ is a counterexample to~$C'$ and we have $w' \leq 2c_4$.
If $c_3>2c_2$, $w' \leq 2c_4 < 2c_4+c_3-2c_2 = c_6$, which contradicts $c_6 \leq w'$.
Therefore $c_3 \leq 2c_2$.
Next, assume that the {\tt +/-} class of $C$ is {\tt +++--+}.
From Theorem~\ref{Thm:tight}, there exist $i$ and $j \ (1 < i \leq j \leq 4)$ such that $c_i+c_j$ is a counterexample to~$C'$.
If $c_3>2c_2$, we have $w' \leq c_i + c_j \leq 2c_4 < c_3+2c_4-2c_2 = c_6$, which contradicts $c_6 \leq w'$.
Hence $c_3 \leq 2c_2$.

As above, we have $2c_2-1 \leq c_3 \leq 2c_2$ and conclude that $C = (1,c_2,2c_2-1,c_4,c_2+c_4-1,2c_4-1)$ or $C = (1,c_2,2c_2,c_4,c_2+c_4,2c_4)$.
\end{proof}

Lemmas~\ref{Lem:2c5-c3_2} and~\ref{Lem:2c5-c3_3} analyze necessary conditions when $C=(1,c_2,2c_2-1,c_4,c_2+c_4-1,2c_4-1)$ and $C=(1,c_2,2c_2,c_4,c_2+c_4,2c_4)$, respectively.

\begin{lemma} \label{Lem:2c5-c3_2}
If $C=(1,c_2,c_3,c_4,c_5,c_6)=(1,c_2,2c_2-1,c_4,c_2+c_4-1,2c_4-1)$ is canonical and the subsystem $C'=(1,c_2,c_3,c_4,c_5)=(1,c_2,2c_2-1,c_4,c_2+c_4-1)$ is noncanonical, then $c_4 \geq 3c_2-1$, $\text{grd}_C(\ell c_3) \leq \ell$, and $\text{grd}_C(\ell c_3) = \ell c_3-c_5+1-\lfloor (\ell c_3-c_5)/c_2 \rfloor (c_2-1)$ for $\ell = \lceil c_5/c_3 \rceil$.
\end{lemma}
\begin{proof}
Assume $c_4 < 3c_2-1$, that is, $c_3+c_5 > c_6$. 
Since $C$ is canonical, $c_3 + c_5$ is not a counterexample to $C$.
Thus, $c_3+c_5-c_6 = 3c_2-c_4-1$ is equal to $1$, $c_2$, $c_3$, or $c_4$.
If $3c_2-c_4-1$ is equal to $c_2$, $c_3$, or $c_4$, then we induce that $c_3 \geq c_4$, which is a contradiction.
If $3c_2-c_4-1$ is equal to $1$, then $C'=(1,c_2,2c_2-1,3c_2-2,4c_2-3)$ and we find that $C'$ is canonical by applying Theorem~\ref{Thm:OPT} repeatedly.
Therefore we have $c_3+c_5 \leq c_6$, namely, $c_4 \geq 3c_2-1$.

Consider paying $\ell c_3$ in~$C$.
Since $c_5 = (c_5 / c_3)\cdot c_3 \leq \lceil c_5/c_3 \rceil \cdot c_3 = \ell c_3$ and $\ell c_3 = \lceil c_5/c_3 \rceil \cdot c_3< c_3+c_5 \leq c_6$, $\text{grd}_{C}^{c_6}(\ell c_3) = 0$ and $\text{grd}_{C}^{c_5}(\ell c_3) = 1$ hold.
In addition, $\text{grd}_{C}^{c_4}(\ell c_3-c_5) = \text{grd}_{C}^{c_3}(\ell c_3-c_5) = 0$ follows from $\ell c_3-c_5 = \lceil c_5/c_3\rceil \cdot c_3 - c_5 < c_3$.
If $\ell c_3-c_5 < c_2$, $\text{grd}_{C}(\ell c_3-c_5) = \text{grd}_{C}^{c_1}(\ell c_3-c_5) = \ell c_3-c_5$ holds, and if $c_2 \leq \ell c_3-c_5 < c_3 = 2c_2-1$, $\text{grd}_{C}^{c_2}(\ell c_3-c_5) = 1$ and $\text{grd}_{C}(\ell c_3-c_5) = 1 + \ell c_3-c_5-c_2$ hold.
Thus we have $\text{grd}_{C}(\ell c_3) = \ell c_3 -c_5+1- \lfloor (\ell c_3-c_5)/c_2 \rfloor (c_2-1)$.
Since $C$ is canonical, $\text{grd}_{C}(\ell c_3) = \text{opt}_{C}(\ell c_3) \leq \ell$.
\end{proof}

\begin{lemma} \label{Lem:2c5-c3_3}
If $C = (1,c_2,c_3,c_4,c_5,c_6) = (1,c_2,2c_2,c_4,c_2+c_4,2c_4)$ is canonical and the subsystem $C'= (1,c_2,c_3,c_4,c_5) = (1,c_2,2c_2,c_4,c_2+c_4)$ is noncanonical, then $c_4 \geq 3c_2-1$, $c_4 \neq 3c_2$, $\text{grd}_{C}(\ell c_3) \leq \ell$, and $\text{grd}_{C}(\ell c_3) = \ell c_3-c_5+1-\lfloor (\ell c_3-c_5)/c_2 \rfloor (c_2-1)$ for $\ell = \lceil c_5/c_3 \rceil$.
\end{lemma}
\begin{proof}
First, assume $c_4 = 3c_2$.
Then $C = (1, c_2, 2c_2, 3c_2, 4c_2, 6c_2)$ holds.
Applying Theorem~\ref{Thm:OPT}, we have that $C'=(1, c_2, 2c_2, 3c_2, 4c_2)$ is canonical, which contradicts the assumption.

Secondly, assume $c_4 > 3c_2$, which is equivalent to $c_6 > c_3+c_5$.
Consider paying $\ell c_3$ in~$C$.
Since $c_5 = (c_5 / c_3)\cdot c_3 \leq \lceil c_5/c_3 \rceil \cdot c_3 = \ell c_3$ and $\ell c_3 = \lceil c_5/c_3 \rceil \cdot c_3< c_3+c_5 < c_6$, $\text{grd}_{C}^{c_6}(\ell c_3) = 0$ and $\text{grd}_{C}^{c_5}(\ell c_3) = 1$ hold.
The remainder of the proof that $\text{grd}_{C}(\ell c_3) \leq \ell$ where $\ell = \lceil c_5/c_3 \rceil$ and $\text{grd}_{C}(\ell c_3) = \ell c_3-c_5+1-\lfloor (\ell c_3-c_5)/c_2 \rfloor (c_2-1)$ for $c_4 > 3c_2$ proceeds in the same way as that of Lemma~\ref{Lem:2c5-c3_2}.

Finally, assume $c_4 < 3c_2$, which is equivalent to $c_6 < c_3+c_5$.
Since $C$ is canonical, $c_3 + c_5$ is not a counterexample to~$C$.
Thus, $c_3 + c_5 - c_6 = 3c_2 - c_4$ is equal to $1$, $c_2$, $c_3$, or~$c_4$.
If $3c_2-c_4$ is equal to $c_2$, $c_3$, or $c_4$, then we induce that $c_3 \geq c_4$, which is a contradiction.
Hence, we have $c_4 = 3c_2-1$ and $C = (1, c_2, 2c_2, 3c_2 -1, 4c_2 -1, 6c_2 -1)$.
Then $\ell$ is $\lceil c_5/c_3 \rceil = 2$ and $\text{grd}_{C}(\ell c_3)$ is $\ell c_3-c_5+1-\lfloor (\ell c_3-c_5)/c_2 \rfloor (c_2-1) = 2$.
Therefore $\text{grd}_{C}(\ell c_3) \leq \ell$ holds.

From the above, we have $c_4 \geq 3c_2 -1$, $c_4 \neq 3c_2$, and $\text{grd}_{C}(\ell c_3) \leq \ell$ where $\ell = \lceil c_5/c_3 \rceil$ and $\text{grd}_{C}(\ell c_3) = \ell c_3-c_5+1-\lfloor (\ell c_3-c_5)/c_2 \rfloor (c_2-1)$.
\end{proof}

Here, we have a necessary condition for a system with six types of coins being canonical and the subsystem with the leading five types of coins being noncanonical.

\begin{thm} \label{Thm:need}
Assume a system $C= (1,c_2,c_3,c_4,c_5,c_6)$ is canonical and the subsystem $C'= (1,c_2,c_3,c_4,c_5)$ is noncanonical.
Then $C$ satisfies (a), (b), or (c) for $\ell = \lceil c_5/c_3 \rceil:$
\begin{itemize}
\item[(a)] $C=(1,2,3,c_4,c_4+1,2c_4)$ and $c_4 > 4;$
\item[(b)] $C=(1,c_2,2c_2-1,c_4,c_2+c_4-1,2c_4-1)$, $c_4 \geq 3c_2-1$, $\text{grd}_C(\ell c_3) = \ell c_3-c_5+1-\lfloor (\ell c_3-c_5)/c_2 \rfloor (c_2-1)$, and $\text{grd}_C(\ell c_3) \leq \ell;$
\item[(c)] $C=(1,c_2,2c_2,c_4,c_2+c_4,2c_4)$, $c_4 \geq 3c_2-1$, $c_4 \neq 3c_2$, $\text{grd}_C(\ell c_3) = \ell c_3-c_5+1-\lfloor (\ell c_3-c_5)/c_2 \rfloor (c_2-1)$, and $\text{grd}_C(\ell c_3) \leq \ell$.
\end{itemize}
\end{thm}
\begin{proof}
This proposition follows from Lemmas~\ref{Lem:c6_kouho}, \ref{Lem:2c5-c4}, \ref{Lem:2c5-c2}, \ref{Lem:2c5-c3_1}, \ref{Lem:2c5-c3_2}, and~\ref{Lem:2c5-c3_3}.
\end{proof}

We now prove the converse of Theorem~\ref{Thm:need}.
The converses of (a), (b), and (c) of Theorem~\ref{Thm:need} correspond to Lemmas~\ref{Lem:g_and_m-1}, \ref{Lem:g_and_m-2}, and \ref{Lem:g_and_m-3}, respectively.

\begin{lemma} \label{Lem:g_and_m-1}
Assume $C=(1,c_2,c_3,c_4,c_5,c_6)=(1,2,3,c_4,c_4+1,2c_4)$ and $c_4 > 4$.
Then $C$ is canonical and the subsystem $C'=(1,c_2,c_3,c_4,c_5)=(1,2,3,c_4,c_4+1)$ is noncanonical.
\end{lemma}
\begin{proof}
The value $2 c_4$ is a counterexample to $C'$ because $\text{opt}_{C'}(2c_4) = 2$ and $\text{grd}_{C'}(2c_4) > 2$, which follows from $2c_4-(c_4+1) = c_4-1 > 3$.
Thus, $C'$ is noncanonical.

We show that $C$ is tight; that is, no counterexample to $C$ exists that is less than or equal to~$c_6$.
Let $C_3$ be the subsystem of $C$ with the leading three types of coins.
Since $C_3 = (1,c_2,c_3)=(1,2,3)$, $C_3$ is canonical.

Consider paying $v$ in $C$ and analyze $\text{grd}_C(v)$.
When $v < c_4$, the equality $\text{grd}_C(v)= \text{grd}_{C_3}(v)$ holds because $v < c_4$.
In addition, $\text{grd}_{C_3}(v)= \text{opt}_{C_3}(v)$ because $C_3$ is canonical.
Hence, if $v < c_4$, $\text{grd}_C(v)= \text{opt}_{C}(v)$ holds and $v$ is not a counterexample to $C$.

Suppose $c_5 < v < c_6$.
Since $v- c_4 < c_4$ and $v-c_5 < c_4-1$, $\text{opt}_C(v)$ is equal to $\text{grd}_{C_3}(v-c_5)+1$, $\text{grd}_{C_3}(v-c_4)+1$, or $\text{grd}_{C_3}(v)$.
As $\text{grd}_{C_3}(v) = \lceil v/3 \rceil$, $\text{grd}_{C_3}(v)$ is monotonically nondecreasing with respect to $v$.
Since $c_4 > 4$, we have $\text{grd}_{C_3}(v-c_5)+1 \leq \text{grd}_{C_3}(v-c_4)+1 \leq \text{grd}_{C_3}(v)$.
Hence, for $c_5 < v < c_6$, $\text{opt}_C(v) = \text{grd}_{C_3}(v-c_5)+1$ holds, which means the greedy algorithm is optimal for $c_5 < v < c_6$.
Therefore, $v$ such that $c_5< v < c_6$ is not a counterexample to~$C$.
Thus, $C$ is tight and accordingly $C'$ is also tight.

From Theorem~\ref{Thm:tight}, there exists a counterexample $w$ to $C$ such that $w = c_i+c_j > c_6 \ (1 < i \leq j \leq 5)$ if $C$ is noncanonical.
Such $w$ can be only $c_4+c_5=c_6+1$ or $c_5 + c_5 = c_6+c_2$, but both of them are not counterexamples to~$C$ because $\text{opt}_C(c_4+c_5) = 2 = \text{grd}_C(c_6+1)$ and $\text{opt}_C(c_5+c_5) = 2 = \text{grd}_C(c_6+c_2)$, and thus $C$ is canonical.
\end{proof}

\begin{lemma} \label{Lem:g_and_m-2}
Assume $C=(1,c_2,c_3,c_4,c_5,c_6)=(1,c_2,2c_2-1,c_4,c_2+c_4-1,2c_4-1)$, $c_4 \geq 3c_2-1$, $\text{grd}_{C}(\ell c_3) = \ell c_3-c_5+1-\lfloor (\ell c_3-c_5)/c_2 \rfloor (c_2-1)$, and $\text{grd}_{C}(\ell c_3) \leq \ell$ for $\ell = \lceil c_5/c_3 \rceil$. 
Then $C$ is canonical and the subsystem $C'=(1,c_2,c_3,c_4,c_5)=(1,c_2,2c_2-1,c_4, c_2+c_4-1)$ is noncanonical.
\end{lemma}

\begin{lemma} \label{Lem:g_and_m-3}
Assume $C = (1,c_2,c_3,c_4,c_5,c_6) = (1,c_2,2c_2,c_4,c_2+c_4,2c_4)$, $c_4 \geq 3c_2-1$, $c_4 \neq 3c_2$, $\text{grd}_{C}(\ell c_3) = \ell c_3-c_5+1-\lfloor (\ell c_3-c_5)/c_2 \rfloor (c_2-1)$, and $\text{grd}_{C}(\ell c_3)\leq \ell$ for $\ell = \lceil c_5/c_3 \rceil$.
Then $C$ is canonical and the subsystem $C' = (1,c_2,c_3,c_4,c_5) = (1,c_2,2c_2,c_4, c_2+c_4)$ is noncanonical.
\end{lemma}

\noindent
The proofs of Lemmas~\ref{Lem:g_and_m-2} and~\ref{Lem:g_and_m-3} are rather long and so can be found in Appendices A and~B, respectively.

We conclude this section with the following theorems, the latter of which coincides with Proposition~\ref{Pro:6coins}.

\begin{thm} \label{Thm:sufficiency}
Let $C=(1,c_2,c_3,c_4,c_5,c_6)$ be a system that satisfies (a), (b), or (c) for $\ell = \lceil c_5/c_3 \rceil$.
Then $C$ is canonical and the subsystem $C'=(1,c_2,c_3,c_4,c_5)$ is noncanonical.
\begin{itemize}
\item[(a)] $C=(1,2,3,c_4,c_4+1,2c_4)$ and $c_4 > 4;$
\item[(b)] $C=(1,c_2,2c_2-1,c_4,c_2+c_4-1,2c_4-1)$, $c_4 \geq 3c_2-1$, $\text{grd}_C(\ell c_3) \leq \ell$, and $\text{grd}_C(\ell c_3) = \ell c_3-c_5+1-\lfloor (\ell c_3-c_5)/c_2 \rfloor (c_2-1);$
\item[(c)] $C=(1,c_2,2c_2,c_4,c_2+c_4,2c_4)$, $c_4 \geq 3c_2-1$, $c_4 \neq 3c_2$, $\text{grd}_C(\ell c_3) \leq \ell$, and $\text{grd}_C(\ell c_3) = \ell c_3-c_5+1-\lfloor (\ell c_3-c_5)/c_2 \rfloor (c_2-1)$.
\end{itemize}
\end{thm}
\begin{proof}
The proposition holds from Lemmas~\ref{Lem:g_and_m-1}, \ref{Lem:g_and_m-2}, and \ref{Lem:g_and_m-3}.
\end{proof}

\begin{thm} \label{Crl:iff2}
A system $C= (1,c_2,c_3,c_4,c_5,c_6)$ is canonical if and only if (a) or (b) holds:
\begin{enumerate}
\item[(a)] the subsystem $(1,c_2,c_3,c_4,c_5)$ is canonical and $\text{grd}_{C}(mc_5) \leq m$ holds for $m = \lceil c_{6} / c_5 \rceil;$
\item[(b)] the subsystem $(1,c_2,c_3,c_4,c_5)$ is noncanonical and $C$ satisfies (i), (ii), or (iii) for $\ell = \lceil c_5/c_3 \rceil$.
In addition, $\text{grd}_C(\ell c_3) = \ell c_3-c_5+1-\lfloor (\ell c_3-c_5)/c_2 \rfloor (c_2-1)$.
\begin{itemize}
\item[(i)] $C=(1,2,3,c_4,c_4+1,2c_4)$ and $c_4 > 4;$
\item[(ii)] $C=(1,c_2,2c_2-1,c_4,c_2+c_4-1,2c_4-1)$, $c_4 \geq 3c_2-1$, and $\text{grd}_C(\ell c_3) \leq \ell;$
\item[(iii)] $C=(1,c_2,2c_2,c_4,c_2+c_4,2c_4)$, $c_4 \geq 3c_2-1$, $c_4 \neq 3c_2$, and $\text{grd}_C(\ell c_3) \leq \ell$.
\end{itemize}
\end{enumerate}
\end{thm}
\begin{proof}
Part (a) comes from Theorem~\ref{Thm:OPT}, and part (b) follows from Theorems \ref{Thm:need} and~\ref{Thm:sufficiency}.
Note that the equation $\text{grd}_C(\ell c_3) = \ell c_3-c_5+1-\lfloor (\ell c_3-c_5)/c_2 \rfloor (c_2-1)$ also holds for $C=(1,2,3,c_4,c_4+1,2c_4)$ and $c_4 > 4$, which can be confirmed via simple calculation.
\end{proof}

\section{Generalization and conclusion} \label{Section4}

This section considers a generalization of the characterization of canonical systems and concludes this study.

The following corollaries for systems with five and six types of coins stem from Theorems \ref{Thm:5coin} and~\ref{Crl:iff2}, respectively.

\begin{corollary} \label{5coins+}
A system $C = (1, c_2, c_3, c_4, 2c_4-c_2)$ is canonical and the subsystem $(1, c_2, c_3, c_4)$ is noncanonical if and only if $C = (1, 2, c_3, c_3 + 1, 2c_3)$ and $c_3 > 3$.
\end{corollary}

\begin{corollary} \label{6coins+}
A system $C = (1, c_2, c_3, c_4, c_5, 2c_5-c_2)$ is canonical and the subsystem $(1, c_2, c_3, c_4, c_5)$ is noncanonical if and only if $C = (1, 2, 3, c_4, c_4+1, 2c_4)$ and $c_4 > 4$.
\end{corollary}

\noindent
Based on the similarity between these corollaries, we arrive at the following theorem that extends them for a general value~$n$.
The proof proceeds in a similar manner to those of Lemmas~\ref{Lem:2c5-c2} and~\ref{Lem:g_and_m-1}.

\begin{thm} \label{Thm:general}
For $n\geq 5$, a system with $n$ types of coins $C = (1, c_2, \ldots, c_{n-1}, 2c_{n-1}-c_2)$ is canonical and the subsystem $C'=(1, c_2, \ldots, c_{n-1})$ is noncanonical if and only if $C = (1, 2, \ldots, n-3, c_{n-2}, c_{n-2}+1, 2c_{n-2})$ and $c_{n-2} > n-2$.
\end{thm}
\begin{proof}
\emph{Necessity}.
Let $w'$ be the minimum counterexample to $C'$.
From Lemma~\ref{Lem:exCor1} and Theorem~\ref{Thm:w_range}, $c_n \leq w' < c_{n-2}+c_{n-1}$.
Thus $\text{grd}_C(c_{n-2}+c_{n-1})>1$, and clearly $\text{opt}_C(c_{n-2}+c_{n-1}) \leq 2$.
Since $C$ is canonical, $\text{grd}_C(c_{n-2}+c_{n-1})= \text{opt}_C(c_{n-2}+c_{n-1})= 2$.
As $c_n < c_{n-2}+c_{n-1}$, we have $\text{grd}_C^{c_n}(c_{n-2}+c_{n-1}) \geq 1$.
Thus there exists $i \in \{1,2,\ldots,n\}$ such that $c_{n-2}+c_{n-1} - c_n = c_{n-2}-c_{n-1}+c_2 = c_i$.
From the last equality, we have $c_2 - c_i = c_{n-1}-c_{n-2}$.
If $i \geq 2$, the equation does not hold and therefore $i=1$.
Hence, we have $c_{n-1} = c_{n-2} + c_2 -1$ and $C=(1,c_2,\ldots,c_{n-2},c_{n-2}+c_2-1,2c_{n-2}+c_2-2)$.

Assume $c_2>2$.
\begin{itemize}
\item[\textbullet] When the last three symbols of the {\tt +/-} class of $C$ are {\tt +-+}. \\
From Theorem~\ref{Thm:OPT}, if $C'$ is noncanonical, $\text{grd}_{C'}(2c_{n-2}) > \text{opt}_{C'}(2c_{n-2})$ holds.
Thus $2c_{n-2}$ is a counterexample to~$C'$, which leads to $w' \leq 2c_{n-2}$.
\item[\textbullet] When the last three symbols of the {\tt +/-} class of $C$ are {\tt --+}. \\
Since $C$ is canonical, the systems $(1,c_2,c_3)$, $(1,c_2,\ldots,c_{n-2})$, and $C'=(1,c_2,\ldots,c_{n-1})$ are tight from Lemma~\ref{Lem:prefix}.
From Theorem~\ref{Thm:pre3}, the system $(1,c_2,c_3)$ is canonical, and from the assumption on the {\tt +/-} class, $(1,c_2,\ldots,c_{n-2})$ and $C'$ are noncanonical.
Thus, from Theorem~\ref{Thm:tight}, there exists a counterexample $c_i + c_j$ to $C'$ such that $1 < i \leq j \leq n-2$.
Therefore we have $w' \leq 2c_{n-2}$.
\end{itemize}
\noindent
As above, $w' \leq 2c_4 < 2c_4+c_2 -2 = c_n$ and this contradicts $c_n \leq w'$.
Hence we have $c_2=2$.

Assume $c_3 > 3$.
Consider paying $c_3 + c_{n-2}$ in $C'$.
Clearly $\text{opt}_{C'}(c_3 + c_{n-2}) \leq 2$.
In addition, $\text{grd}_{C'}^{c_{n-1}}(c_3 + c_{n-2}) = 1$ because $c_3 + c_{n-2} - c_{n-1} = c_3 - 1$.
Since $c_3 + c_{n-2} - c_{n-1}= c_3 - 1 > 2$ and $c_2=2$, we have $\text{grd}_{C'}(c_3 + c_{n-2}) \geq 3 > \text{opt}_{C'}(c_3 + c_{n-2})$ and $c_3 + c_{n-2}$ is a counterexample to $C'$.
Therefore $w' \leq c_3 + c_{n-2} < 2 c_{n-2} = c_n$; however, we already have $c_n \leq w'$.
Hence $c_3 = 3$ holds.

Applying the same argument, we can induce that $c_j=j$ for $j=4,5,\ldots,n-3$.
Hence, $C = (1,2,\ldots,n-3,c_{n-2},c_{n-2}+1,2c_{n-2})$.

Assume $c_{n-2} = n-2$.
Then $C ' = (1,2,\ldots,n-3,n-2,n-1)$ and $C'$ is canonical from Theorem~\ref{Thm:OPT}.
Thus $C = (1,2,\ldots,n-3,c_{n-2},c_{n-2}+1,2c_{n-2})$ and $c_{n-2}>n-2$. \\

\noindent
\emph{Sufficiency}.
The value $2 c_{n-2}$ is a counterexample to $C'$ because $\text{opt}_{C'}(2c_{n-2}) = 2$ and $\text{grd}_{C'}(2c_{n-2}) > 2$, which follows from $2c_{n-2}-(c_{n-2}+1) = c_{n-2} - 1 > n-3$.
Thus $C'$ is noncanonical.

Set $c_{n-1} := c_{n-2}+1$ and $c_n := 2c_{n-2}$.
We show that $C$ is tight; that is, no counterexample to $C$ exists that is less than or equal to $c_n$.
Denote the subsystem $(1,c_2,\ldots,c_{n-3})$ of $C$ by $C_{n-3}$.
Since $C_{n-3} =(1,2,\ldots,n-3)$, $C_{n-3}$ is canonical.

Consider paying $v$ in $C$ and analyze $\text{grd}_C(v)$.
When $v < c_{n-2}$, the equality $\text{grd}_C(v)= \text{grd}_{C_{n-3}}(v)$ holds.
In addition, $\text{grd}_{C_{n-3}}(v)= \text{opt}_{C_{n-3}}(v)$ because $C_{n-3}$ is canonical.
Hence, if $v < c_{n-3}$, $\text{grd}_C(v)= \text{opt}_{C}(v)$ holds and $v$ is not a counterexample to $C$.

Suppose $c_{n-1} < v < c_n$.
Since $v- c_{n-2} < c_{n-2}$ and $v-c_{n-1} < c_{n-2}-1$, $\text{opt}_C(v)$ is equal to $\text{grd}_{C_{n-3}}(v-c_{n-1})+1$, $\text{grd}_{C_{n-3}}(v-c_{n-2})+1$, or $\text{grd}_{C_{n-3}}(v)$.
As $\text{grd}_{C_{n-3}}(v) = \lceil v/(n-3) \rceil$, $\text{grd}_{C_{n-3}}(v)$ is monotonically nondecreasing with respect to $v$.
Since $c_{n-2} > n-2$, we have $\text{grd}_{C_{n-3}}(v-c_{n-1})+1 \leq \text{grd}_{C_{n-3}}(v-c_{n-2})+1 \leq \text{grd}_{C_{n-3}}(v)$.
Hence, for $c_{n-1} < v < c_n$, $\text{opt}_C(v) = \text{grd}_{C_{n-3}}(v-c_{n-1})+1$ holds, which means the greedy algorithm is optimal for $c_{n-1} < v < c_n$.
Therefore $v$~such that $c_{n-1}< v < c_n$ is not a counterexample to $C$.
Thus $C$ is tight and accordingly $C'$ is also tight.

From Theorem~\ref{Thm:tight}, there exists a counterexample $w$ to $C$ such that $w = c_i+c_j > c_n \ (1 < i \leq j \leq n-1)$ if $C$ is noncanonical.
Such $w$ can be only $c_{n-2}+c_{n-1} = c_n+1$ or $2c_{n-2}=c_n+c_2$, but neither of them is a counterexample to~$C$ because $\text{opt}_C(c_{n-2}+c_{n-1}) = 2 = \text{grd}_C(c_n+1)$ and $\text{opt}_C(2c_{n-2}) = 2 = \text{grd}_C(c_n+c_2)$, and thus $C$ is canonical.
\end{proof}

From an argument similar to that for Lemma~\ref{Lem:c6_kouho}, if $C=(1,c_2,\ldots,c_n)$ is canonical and $C'=(1,c_2,\ldots,c_{n-1})$ is noncanonical, then $c_n$ is equal to $2c_{n-1}-c_2$, $2c_{n-1}-c_2$, $\ldots$, or $2c_{n-1}-c_{n-2}$.
Theorem~\ref{Thm:general} covers one of them, namely, $c_n=2c_{n-1}-c_2$.

In this paper, we have provided characterization of canonical systems with six types of coins for the change-making problem.
Moreover, we have proposed a partial characterization of canonical systems with more than six types of coins.
In future work, we plan to extend the characterization and theorems obtained in this study to a general~case.

%



\appendix
\section{Proof of Lemma~\ref{Lem:g_and_m-2}} \label{appendA}

This appendix describes the proof of Lemma~\ref{Lem:g_and_m-2}, which states the following proposition.
\begin{quote}
Assume $C=(1,c_2,c_3,c_4,c_5,c_6)=(1,c_2,2c_2-1,c_4,c_2+c_4-1,2c_4-1)$, $c_4 \geq 3c_2-1$, $\text{grd}_{C}(\ell c_3) = \ell c_3-c_5+1-\lfloor (\ell c_3-c_5)/c_2 \rfloor (c_2-1)$, and $\text{grd}_{C}(\ell c_3) \leq \ell$ for $\ell = \lceil c_5/c_3 \rceil$. 
Then $C$ is canonical and the subsystem $C'=(1,c_2,c_3,c_4,c_5)=(1,c_2,2c_2-1,c_4, c_2+c_4-1)$ is noncanonical.
\end{quote}
The proof is slightly similar to but more complicated than that of Lemma~\ref{Lem:g_and_m-3}, which is given in Appendix~B.

\begin{proof}
The value $2 c_4$ is a counterexample to $C'$ because $\text{opt}_{C'}(2c_4) = 2$ and $\text{grd}_{C'}(2c_4) > 2$, which is shown as follows.
Since $c_4 \geq 3c_2-1$ holds from the assumption, we have $2c_4 -c_5 = c_4-c_2+1 \geq 2c_2 > c_3$, and thus $\text{grd}_{C'}(2c_4) = 1 + \text{grd}_{C'}(2c_4 -c_5) = 1 + \text{grd}_{C'}(c_4-c_2+1) > 2$.
Hence $2 c_4$ is a counterexample to $C'$, and $C'$ is noncanonical.

We show that $C$ is tight, that is, no counterexample to $C$ exists that is less than or equal to~$c_6$.
Consider paying $v$ and analyze $\text{grd}_C(v)$.
Let $C_3$ be the subsystem of $C$ with the leading three types of coins: $C_3 = (1,c_2,c_3)=(1,c_2,2c_2-1)$.
When $v < c_4$, $\text{grd}_C(v)= \text{grd}_{C_3}(v)$ holds because $v < c_4$ and $C_3$ is canonical.
Thus, no counterexample to $C$ exists less than or equal to $c_4$. \\

Suppose $c_4 < v < c_5$.
Then $\text{grd}_{C}(v) = \text{grd}_{C_3}(v-c_4)+1$ holds.
The value $\text{opt}_C(v)$ is equal to $\text{grd}_{C_3}(v-c_4)+1$ or $\text{grd}_{C_3}(v)$.
We prove that $\text{opt}_C(v) = \text{grd}_{C_3}(v-c_4)+1 = \text{grd}_{C}(v)$ for $c_4 < v < c_5$ by showing $\text{grd}_{C_3}(v-c_4)+1 \leq \text{grd}_{C_3}(v)$.

Without loss of generality, $c_4$ can be represented as $c_4 = 2c_2 + sc_3 + t$ for $s \in \mathbb{Z}_{\geq 0}$ and $0 \leq t < c_3$.
By using this representation, $\ell$, $\ell c_3-c_5$, and $\text{grd}_{C}(\ell c_3)$ are calculated as follows:
\begin{align}
\ell & = \lceil c_5 / c_3 \rceil \notag \\
& = 
\begin{cases}
\ s+2 & (0 \leq t < c_2) \\
\ s+3 & (c_2 \leq t < c_3)
\end{cases}, \label{EQA02} \\
\ell c_3 - c_5 &= 
\begin{cases}
\ c_2-t-1  & (0 \leq t < c_2) \\
\ 3c_2-t-2 & (c_2 \leq t < c_3)
\end{cases}, \notag \\
\text{grd}_{C}(\ell c_3) &= \ell c_3-c_5+1-\lfloor (\ell c_3-c_5)/c_2 \rfloor (c_2-1) \notag \\
& = 
\begin{cases}
\ c_2-t   & (0 \leq t < c_2) \\
\ c_3-t+1 & (c_2 \leq t < c_3)
\end{cases}. \label{EQA04}
\end{align}
From \eqref{EQA02} and \eqref{EQA04}, the following relationship holds:
\begin{align}
\text{grd}_C(\ell c_3) \leq \ell \quad \Longleftrightarrow \quad
\begin{cases}
\ s+t+2 \geq c_2 & (0 \leq t < c_2) \\
\ s+t+2 \geq c_3 & (c_2 \leq t < c_3)
\end{cases}. \label{EQA05}
\end{align}
In particular, we have
\begin{align}
\text{grd}_C(\ell c_3) \leq \ell \quad \Longleftrightarrow \quad s+t+2 \geq c_2  \quad (0 \leq t < c_3). \label{EQA05-2}
\end{align}
Since $c_4 < v < c_5$, $v$ can be represented as $v = c_4+u$ where $0 \leq u < c_2-1$.
Then,
\begin{align}
\text{grd}_{C_3}(v-c_4)+1 &= u+ 1 \label{EQA06}
\end{align}
and
\begin{align}
\text{grd}_{C_3}(v)       &= \text{grd}_{C_3}(c_4+u) \notag \\
&= \text{grd}_{C_3}(2c_2+sc_3+t+u) \notag \\
&= s+1 + \text{grd}_{C_3}(t+u+1) \notag
\end{align}
hold.
Since $0 \leq t < c_3$ and $0 \leq u < c_2 -1$, $t+u+1$ is less than $c_2 + c_3 -1$.
Hence,
\begin{align}
\text{grd}_{C_3}(v)       &= s+1 + \text{grd}_{C_3}(t+u+1) \notag \\
&=
\begin{cases}
\ (s+t+1) + (u+1)    & (0 < t+u+1 < c_2) \\
\ (s+t+2-c_2)+(u+1)  & (c_2 \leq t+u+1 < c_3) \\
\ (s+t+2-c_3)+(u+1)  & (c_3 \leq t+u+1 < c_2 + c_3 - 1) 
\end{cases}. \label{EQA07}
\end{align}
If $0 < t+u+1 < c_2$,
\begin{align}
\text{grd}_{C_3}(v-c_4)+1 \leq \text{grd}_{C_3}(v) \label{EQA09}
\end{align}
holds from \eqref{EQA06} and \eqref{EQA07};
if $c_2 \leq t+u+1 < c_3$, the inequality \eqref{EQA09} holds from \eqref{EQA05-2}, \eqref{EQA06}, and~\eqref{EQA07};
if $c_3 \leq t+u+1 < c_2 + c_3 - 1$, the inequality \eqref{EQA09} holds from \eqref{EQA05}, \eqref{EQA06}, and~\eqref{EQA07} because $c_2 \leq t$ when $c_3 \leq t+u+1$.
Therefore we have $\text{grd}_{C_3}(v-c_4)+1 \leq \text{grd}_{C_3}(v)$, which implies $\text{opt}_C(v) =\text{grd}_{C_3}(v-c_4)+1 = \text{grd}_C(v)$, and thus~$v \ (c_4 < v < c_5)$ is not a counterexample to~$C$.
\\

Suppose $c_5 < v < c_6$.
Then, $\text{grd}_{C}(v) = \text{grd}_{C_3}(v-c_5)+1$ holds.
The value $\text{opt}_C(v)$ is equal to $\text{grd}_{C_3}(v-c_5)+1$, $\text{grd}_{C_3}(v-c_4)+1$, or $\text{grd}_{C_3}(v)$.
We prove that $\text{opt}_C(v) = \text{grd}_{C_3}(v-c_5)+1 = \text{grd}_{C}(v)$ for $c_5 < v < c_6$ by showing $\text{grd}_{C_3}(v-c_5)+1 \leq \text{grd}_{C_3}(v-c_4)+1$ and $\text{grd}_{C_3}(v-c_5)+1 \leq \text{grd}_{C_3}(v)$.

Firstly, we prove $\text{grd}_{C_3}(v-c_5)+1 \leq \text{grd}_{C_3}(v-c_4)+1$.
Since $c_5 < v < c_6$, we have $0 < v-c_5 < c_4 - c_2$, which can be divided into $0 < v-c_5 < c_2$, $c_2 \leq v-c_5 < c_3$, and $c_3 \leq v-c_5 < c_4-c_2$.
We consider these three cases in the following.

If $0 < v-c_5 < c_2$, which is equivalent to $c_2-1 < v-c_4 < c_3$, the following relationships hold:
\begin{align*}
\text{grd}_{C_3}(v-c_5) &= v-c_5, \\
\text{grd}_{C_3}(v-c_4) &= v-c_4-c_2+1 \\
&= v- c_5.
\end{align*}
Thus we have $\text{grd}_{C_3}(v-c_5)+1 \leq \text{grd}_{C_3}(v-c_4)+1$ when $0 < v-c_5 < c_2$.

If $c_2 \leq v-c_5 < c_3$, which is equivalent to $c_3 < v-c_4 < c_2 + c_3 - 1$, the following relationships hold:
\begin{align*}
\text{grd}_{C_3}(v-c_5) &= v - c_5 - c_2 + 1, \\
\text{grd}_{C_3}(v-c_4) &= 1 + \text{grd}_{C_3}(v-c_4-c_3) \\
                        &= v - c_4 - 2c_2 + 2 \\
                        &= v - c_5 - c_2 + 1.
\end{align*}
Hence, $\text{grd}_{C_3}(v-c_5)+1 \leq \text{grd}_{C_3}(v-c_4)+1$ holds when $c_2 \leq v-c_5 < c_3$.

If $c_3 \leq v-c_5 < c_4 - c_2$, it is equivalent to $c_2+c_3-1 \leq v-c_4 < c_4 - 1$.
Assume $v-c_5$ is equal to $pc_3 + q$ where $p \in \mathbb{Z}_{> 0}$ and $0 \leq q < c_3$.
Then, $v - c_4 = pc_3 + c_2 + q-1$ holds.
In addition, we have
\begin{align*}
\text{grd}_{C_3}(v-c_5) &= p + \text{grd}_{C_3}(q) \\
&=
\begin{cases}
\ p          & (q=0) \\
\ p+q        & (0<q<c_2) \\
\ p+q-c_2+1  & (c_2\leq q<c_3)
\end{cases},\\
\text{grd}_{C_3}(v-c_4) &= \text{grd}_{C_3}(pc_3 + c_2 + q - 1) \\
&= p + \text{grd}_{C_3}(c_2 + q - 1) \\
&= 
\begin{cases}
\ p+c_2-1    & (q=0) \\
\ p+q        & (0<q<c_2) \\
\ p+q-c_2+1  & (c_2\leq q<c_3)
\end{cases}.
\end{align*}
Hence, $\text{grd}_{C_3}(v-c_5)+1 \leq \text{grd}_{C_3}(v-c_4)+1$ holds when $c_3 \leq v-c_5 < c_4 - c_2$.
We therefore obtain $\text{grd}_{C_3}(v-c_5)+1 \leq \text{grd}_{C_3}(v-c_4)+1$ for $c_5 < v < c_6$.

Secondly, we prove $\text{grd}_{C_3}(v-c_5)+1 \leq \text{grd}_{C_3}(v)$ when $c_5 < v < c_6$.
Let $D(v)$ be $D(v) = \text{grd}_{C_3}(v) - (\text{grd}_{C_3}(v-c_5)+1)$.
To show $\min_{c_5 < v < c_6}D(v) \geq 0$, we prove (a)--(e) in order:
\begin{itemize}
\item[(a)] $\min_{c_5 < v < c_6}D(v) = \min_{c_5 < v < c_5+c_3}D(v);$
\item[(b)] for $c_5 < v < c_5 + c_2 -1$, $D(c_5+c_2-1) \leq D(v);$
\item[(c)] for $c_5 + c_2 - 1 < v < c_5 + c_3 -1$, $D(c_5+c_3-1) \leq D(v);$
\item[(d)] $D(c_5+c_2-1) \leq D(c_5+c_3-1);$
\item[(e)] $D(c_5+c_2-1) \geq 0$.
\end{itemize}

(a).
The following relationships obviously hold:
\begin{align}
\text{grd}_{C_3}(v+ c_3)      &= \text{grd}_{C_3}(v) + 1,      \label{EQB01} \\
\text{grd}_{C_3}(v-c_5 + c_3) &= \text{grd}_{C_3}(v-c_5) + 1.  \label{EQB02} 
\end{align}
Subtracting \eqref{EQB02} from \eqref{EQB01}, we have
\begin{align}
D(v) = D(v+c_3). \label{EQB03}
\end{align}
Since $c_4 \geq 3c_2-1$, we have $c_6-c_5 = c_4-c_2 \geq c_3$, that is, $c_5 + c_3 \leq c_6$.
Thus,
\begin{align}
\min_{c_5 < v < c_6}D(v) = \min_{c_5 < v < c_5+c_3}D(v) \notag
\end{align}
holds from the periodicity \eqref{EQB03}.

(b).
We show that for $c_5 < v < c_5 + c_2 -1$, $D(c_5+c_2-1) \leq D(v)$.
The value $v \ (c_5 < v < c_5 + c_2 -1)$ can be represented as $v = c_5+c_2-1 -r$ for $0 < r < c_2-1$.
Then,
\begin{align}
D(c_5+c_2-1) & = \text{grd}_{C_3}(c_5+c_2-1) - \text{grd}_{C_3}(c_2-1) - 1 \notag \\
& = \text{grd}_{C_3}(c_5+c_2-1) - c_2, \label{EQB04} \\
D(c_5+c_2-1-r) & = \text{grd}_{C_3}(c_5+c_2-1-r) - \text{grd}_{C_3}(c_2-1-r) -1  \notag \\
& = \text{grd}_{C_3}(c_5+c_2-1-r) - c_2+r. \label{EQB05}
\end{align}
Assume $c_5+c_2-1 = p c_3 + q$ where $ p \in \mathbb{Z}_{> 0}$ and $0 \leq q < c_3$.
Then $\text{grd}_{C_3}(c_5+c_2-1)$ and $\text{grd}_{C_3}(c_5+c_2-1-r)$ are calculated as follows:
\begin{align}
\text{grd}_{C_3}(c_5+c_2-1) & = p + \text{grd}_{C_3}(q) \notag \\
& = \begin{cases}
\ p + q     & (0 \leq q < c_2) \\
\ p+q-c_2+1 & (c_2 \leq q < c_3) \\
\end{cases}, \label{EQB06}\\
\text{grd}_{C_3}(c_5+c_2-1-r) & = \notag
\begin{cases}
\ p+ \text{grd}_{C_3}(q-r) & (r \leq q) \\
\ p-1 + \text{grd}_{C_3}(c_3+q-r) & (r > q) \\
\end{cases} \\
& = \begin{cases}
\ p+q-r & (r \leq q \text{ and } q-r<c_2) \\
\ p+q-r-c_2+1 & (r \leq q \text{ and } q-r \geq c_2) \\
\ p+q-r+c_2-1 & (r>q)
\end{cases}.\label{EQB07}
\end{align}
Rearranging \eqref{EQB04}--\eqref{EQB07}, we have
\begin{align}
D(c_5+c_2-1) &= 
\begin{cases}
\ p + q - c_2 &(0   \leq q < c_2) \\
\ p + q - c_3 &(c_2 \leq q < c_3) \\
\end{cases}, \notag \\
D(c_5+c_2-1-r) &= 
\begin{cases}
\ p + q - c_2 & (r \leq q \text{ and }  q-r< c_2) \\
\ p + q - c_3 & (r \leq q \text{ and }  q-r \geq c_2) \\
\ p + q - 1   & (r > q) \\
\end{cases}. \notag
\end{align}
When $D(c_5+c_2-1-r)=p+q-c_3$, we have $q-r \geq c_2$, that is, $q \geq c_2 +r$, and thus $D(c_5+c_2-1) = p + q - c_3 = D(c_5+c_2-1-r)$; otherwise, $D(c_5+c_2-1-r)$ is at least $p+q-c_2$, whereas $D(c_5+c_2-1)$ is at most $p+q-c_2$.
Hence, $D(c_5+c_2-1) \leq D(v)$ holds for $c_5 < v < c_5 + c_2 -1$.

(c).
As in (b), we show that $D(c_5+c_3-1) \leq D(v)$ for $c_5+c_2-1 < v < c_5+c_3-1$.
The value $v \ (c_5 + c_2 -1 < v < c_5+c_3-1)$ can be represented as $v = c_5+c_3-1 -r'$ for $0 < r' < c_2-1$.
Then,
\begin{align}
D(c_5+c_3-1) & = \text{grd}_{C_3}(c_5+c_3-1) - \text{grd}_{C_3}(c_3-1) - 1 \notag \\
& = \text{grd}_{C_3}(c_5+c_3-1) - c_2, \label{EQB08} \\
D(c_5+c_3-1-r') & = \text{grd}_{C_3}(c_5+c_3-1-r') - \text{grd}_{C_3}(c_3-1-r') -1  \notag \\
& = \text{grd}_{C_3}(c_5+c_3-1-r') - c_2+r'. \label{EQB09}
\end{align}
Assume $c_5+c_3-1 = p' c_3 + q'$ where $ p' \in \mathbb{Z}_{> 0}$ and $0 \leq q' < c_3$.
Then $\text{grd}_{C_3}(c_5+c_3-1)$ and $\text{grd}_{C_3}(c_5+c_3-1-r')$ are calculated as follows:
\begin{align}
\text{grd}_{C_3}(c_5+c_3-1) & = p' + \text{grd}_{C_3}(q') \notag \\
& = \begin{cases}
\ p' + q'     & (0 \leq q' < c_2) \\
\ p'+q'-c_2+1 & (c_2 \leq q' < c_3) \\
\end{cases}, \label{EQB10}\\
\text{grd}_{C_3}(c_5+c_3-1-r') & = \notag
\begin{cases}
\ p'+ \text{grd}_{C_3}(q'-r')        & (r' \leq q') \\
\ p'-1 + \text{grd}_{C_3}(c_3+q'-r') & (r' > q') \\
\end{cases} \\
& = \begin{cases}
\ p'+q'-r' & (r' \leq q' \text{ and }  q'-r'<c_2) \\
\ p'+q'-r'-c_2+1 & (r' \leq q' \text{ and }  q'-r' \geq c_2) \\
\ p'+q'-r'+c_2-1 & (r'>q')
\end{cases}.\label{EQB11}
\end{align}
Rearranging \eqref{EQB08}--\eqref{EQB11}, we have
\begin{align}
D(c_5+c_3-1) &= 
\begin{cases}
\ p' + q' - c_2 &(0   \leq q' < c_2) \\
\ p' + q' - c_3 &(c_2 \leq q' < c_3) \\
\end{cases}, \notag \\
D(c_5+c_3-1-r') &= 
\begin{cases}
\ p' + q' - c_2 & (r' \leq q' \text{ and } q'-r'< c_2) \\
\ p' + q' - c_3 & (r' \leq q' \text{ and } q'-r' \geq c_2) \\
\ p' + q' - 1   & (r' > q') \\
\end{cases}.\notag
\end{align}
When $D(c_5+c_3-1-r')=p'+q'-c_3$, we have $q'-r' \geq c_2$, that is, $q' \geq c_2 +r'$, and thus $D(c_5+c_3-1) = p' + q' - c_3 = D(c_5+c_3-1-r')$;
otherwise, $D(c_5+c_3-1-r')$ is at least $p'+q'-c_2$, whereas $D(c_5+c_3-1)$ is at most $p'+q'-c_2$.
Hence, $D(c_5+c_3-1) \leq D(v)$ holds for $c_5 + c_2 - 1< v < c_5 + c_3 -1$.

(d).
We show $D(c_5+c_2-1) \leq D(c_5+c_3-1)$.
From \eqref{EQB04} and \eqref{EQB08}, we have $D(c_5+c_2-1) = \text{grd}_{C_3}(c_5+c_2-1) - c_2$ and $D(c_5+c_3-1) = \text{grd}_{C_3}(c_5+c_3-1) - c_2$, respectively.
Assume $c_5+c_2-1 = p c_3 + q$ where $ p \in \mathbb{Z}_{> 0}$ and $0 \leq q < c_3$.
Then $\text{grd}_{C_3}(c_5+c_2-1)$ is given by \eqref{EQB06}, and by using $p$ and $q$, $\text{grd}_{C_3}(c_5+c_3-1)$ is calculated as follows:
\begin{align}
\text{grd}_{C_3}(c_5+c_2-1)
& =
\begin{cases}
\ p + q     & (0 \leq q < c_2) \\
\ p+q-c_2+1 & (c_2 \leq q < c_3) \\
\end{cases},  \label{EQB14} \\
\text{grd}_{C_3}(c_5+c_3-1) &= \text{grd}_{C_3}(pc_3+q+c_2-1) \notag \\
& =
\begin{cases}
\ p + \text{grd}_{C_3}(q+c_2-1) & (0 \leq q < c_2) \\
\ p+1 + \text{grd}_{C_3}(q-c_2)                     & (c_2 \leq q < c_3) \\
\end{cases} \notag \\
&=
\begin{cases}
\ p + c_2-1 & (q = 0) \\
\ p + q                     & (0 < q < c_2) \\
\ p+q -c_2 +1               & (c_2 \leq q < c_3) \\
\end{cases}. \label{EQB15}
\end{align}
From \eqref{EQB14} and \eqref{EQB15}, $\text{grd}_{C_3}(c_5+c_2-1) \leq \text{grd}_{C_3}(c_5+c_3-1)$ when $q=0$; otherwise $\text{grd}_{C_3}(c_5+c_2-1) = \text{grd}_{C_3}(c_5+c_3-1)$.
Therefore we obtain $D(c_5+c_2-1) \leq D(c_5+c_3-1)$.

(e).
We show $D(c_5+c_2-1) \geq 0$.
The value $c_4$ can be represented as $c_4 = 2c_2 + sc_3 + t$ for $s \in \mathbb{Z}_{\geq 0}$ and $0 \leq t < c_3$.
Using this representation, we have
\begin{align}
D(c_5+c_2-1) &= \text{grd}_{C_3}(c_5 + c_2 -1) - \text{grd}_{C_3}(c_2 -1) - 1 \notag \\
&= \text{grd}_{C_3}(c_5 + c_2 -1) - c_2 \notag \\
&= \text{grd}_{C_3}((s+2)c_3 + t) - c_2 \notag \\
&= s+2 + \text{grd}_{C_3}(t) - c_2 \notag \\
&=
\begin{cases}
\ s+t+2-c_2 & (0 \leq t < c_2)   \\
\ s+t+2-c_3 & (c_2 \leq t < c_3)
\end{cases}. \notag
\end{align}
Note that the relationship~\eqref{EQA05} holds not only for $c_4<v<c_5$ but also for $c_5<v<c_6$.
Thus, we have $D(c_5+c_2-1) \geq 0$.

From (a)--(e), we conclude that $\text{grd}_{C_3}(v-c_5)+1 \leq \text{grd}_{C_3}(v)$ and $v$ is not a counterexample to~$C$ when $c_5 < v < c_6$.
\\

From the above discussion, we conclude that $C$ is tight, which directly implies that $C'$ is also tight.
In addition, $C_3$ is canonical and $C'$ is noncanonical.
From Theorem~\ref{Thm:tight}, if $C$ is noncanonical, then there exist $i$ and $j$ such that $1 < i \leq j \leq 5$, $c_i + c_j > c_{6} = 2c_4-1$, and $c_i + c_j$ is a counterexample to~$C$.
The pairs $(i,j) = (4,4)$, $(4,5)$, and $(5,5)$ can be such ones.
The equality $\text{opt}_C(c_4+c_4) = \text{opt}_C(c_4+c_5) = \text{opt}_C(c_5+c_5) = 2$ holds because $c_4 + c_4 > c_6$, $c_4 + c_5 > c_6$, and $c_5 + c_5 > c_6$.
On the other hand, $\text{grd}_C(c_4+c_4) = \text{grd}_C(c_4+c_5) = \text{grd}_C(c_5+c_5) = 2$ holds because $c_4 + c_4 = c_6 + 1$, $c_4 + c_5 = c_2 + c_6$, and $c_5 + c_5 = c_3 + c_6$.
Thus, $c_4+c_4$, $c_4+c_5$, and $c_5+c_5$ are not counterexamples to $C$, and therefore $C$ is canonical.
\end{proof}

\section{Proof of Lemma~\ref{Lem:g_and_m-3}} \label{appendB}

This appendix describes the proof of Lemma~\ref{Lem:g_and_m-3}, which states the following proposition.
\begin{quote}
Assume $C = (1,c_2,c_3,c_4,c_5,c_6) = (1,c_2,2c_2,c_4,c_2+c_4,2c_4)$, $c_4 \geq 3c_2-1$, $c_4 \neq 3c_2$, $\text{grd}_{C}(\ell c_3) = \ell c_3-c_5+1-\lfloor (\ell c_3-c_5)/c_2 \rfloor (c_2-1)$, and $\text{grd}_{C}(\ell c_3)\leq \ell$ for $\ell = \lceil c_5/c_3 \rceil$.
Then $C$ is canonical and the subsystem $C' = (1,c_2,c_3,c_4,c_5) = (1,c_2,2c_2,c_4, c_2+c_4)$ is noncanonical.
\end{quote}

\begin{proof}
The value $c_6 = 2 c_4$ is a counterexample to $C'$ because $\text{opt}_{C'}(2c_4) = 2$ and $\text{grd}_{C'}(2c_4) > 2$, which is shown as follows.
From the assumption, $c_4 = 3c_2 -1$ or $c_4 \geq 3c_2+1$.
When $c_4 = 3 c_2 -1$, $\text{grd}_{C'}(2c_4) = \text{grd}_{C'}(2c_4 - c_5) + 1= \text{grd}_{C'}(2c_4 - (c_2+c_4)) +1 = \text{grd}_{C'}(2 c_2-1) +1 > 2$.
When $c_4 \geq 3 c_2 +1$, $\text{grd}_{C'}(2c_4) = \text{grd}_{C'}(2c_4-c_5)+1=\text{grd}_{C'}(2c_4 - (c_2+c_4)) +1= \text{grd}_{C'}(c_4-c_2)+1$.
Since $c_4-c_2 \geq 2c_2+1 = c_3+1$, $\text{grd}_{C'}(c_4-c_2) > 1$ and thus $\text{grd}_{C'}(2c_4)>2$.
Hence, $C'$ is noncanonical.

We show that $C$ is tight; that is, no counterexample to $C$ exists that is less than or equal to~$c_6$. 
Let $C_3$ be the subsystem of $C$ with the leading three types of coins.
Since $C_3 = (1,c_2,c_3)=(1,c_2,2c_2)$, $C_3$ is canonical.

Consider paying $v$ in $C$ and analyze $\text{grd}_C(v)$.
When $v < c_4$, the equality $\text{grd}_C(v)= \text{grd}_{C_3}(v)$ holds because $v < c_4$.
In addition, $\text{grd}_{C_3}(v)= \text{opt}_{C_3}(v)$ because $C_3$ is canonical.
Hence, if $v < c_4$, $\text{grd}_C(v)= \text{opt}_{C}(v)$ holds and $v$ is not a counterexample to $C$.
\\

Suppose $c_4 < v < c_5$.
Then, $\text{grd}_{C}(v) = \text{grd}_{C_3}(v-c_4)+1$ holds because $v-c_4 < c_2$.
The value $\text{opt}_C(v)$ is equal to $\text{grd}_{C_3}(v-c_4)+1$ or $\text{grd}_{C_3}(v)$.
We prove that $\text{opt}_C(v) = \text{grd}_{C_3}(v-c_4)+1 = \text{grd}_{C}(v)$ and $v$ is not a counterexample to $C$ by showing $\text{grd}_{C_3}(v-c_4)+1 \leq \text{grd}_{C_3}(v)$.

Assume that $v\ (c_4 < v < c_5)$ is the minimum counterexample to $C$.
From Lemma~\ref{Lem:coin1}, we can set $v=pc_2+qc_4$ where $p \in \mathbb{Z}_{> 0}$ and $q \in \{0,1\}$.
Since $c_4 < v < c_5=c_2+c_4$, $q$ is equal to zero when $c_4 < v < c_5$.
Thus we have $v=pc_2$~and
\begin{align}
\text{grd}_{C_3}(v) &= \text{grd}_{C_3}(pc_2) \notag \\
                    &=
\begin{cases}
\ p/2 & (p \text{ is even}) \\
\ (p-1)/2+1 & (p \text{ is odd})
\end{cases}.\label{EQ1}
\end{align}
Without loss of generality, we may assume $c_4 = sc_2 - t$ where $s \in \mathbb{Z}_{>0}$ and $0 \leq t < c_2$.
When $c_4 < v < c_5$, we have $c_4 < v=pc_2 < c_5 = c_4 + c_2$ and therefore $s=p$ holds.
Hence,
\begin{align}
\text{grd}_{C_3}(v-c_4) &= \text{grd}_{C_3}(pc_2-sc_2+t) \notag \\
                        &= \text{grd}_{C_3}(t) \notag \\
                        &= t. \label{EQ2}
\end{align}
In addition, since $\ell = \lceil c_5/c_3 \rceil = \lceil (c_2+c_4)/2c_2\rceil = \lceil (sc_2 + c_2 - t)/2c_2\rceil = \lceil (pc_2 + c_2 - t)/2c_2\rceil$,
\begin{align}
\ell = \begin{cases}
\ p/2+1     & (p \text{ is even})    \\
\ (p-1)/2+1 & (p \text{ is odd}) 
\end{cases}.\label{EQ3}
\end{align}
Summarizing $\text{grd}_C(\ell c_3) = \ell c_3-c_5+1-\lfloor (\ell c_3-c_5)/c_2 \rfloor (c_2-1)$, $c_5 = c_2+c_4$, $c_4 = sc_2-t = pc_2-t$, and $c_3=2c_2$, we~have
\begin{align}
\text{grd}_{C}(\ell c_3) = \begin{cases}
\ t+2 & (p \text{ is even}) \\
\ t+1 & (p \text{ is odd})
\end{cases}.\label{EQ4}
\end{align}
From \eqref{EQ1}--\eqref{EQ4},
\begin{align*}
\text{grd}_{C_3}(v) - (\text{grd}_{C_3}(v-c_4) + 1) &= \begin{cases}
\ (p/2)-(t+1)       & (p \text{ is even}) \\
\ ((p-1)/2+1)-(t+1) & (p \text{ is odd})
\end{cases}\\
&= \ell-\text{grd}_C(\ell c_3).
\end{align*}
Since $\text{grd}_C(\ell c_3) \leq \ell$ holds, we have $\text{grd}_{C_3}(v) - (\text{grd}_{C_3}(v-c_4) + 1) \geq 0$.

As mentioned before, for $c_4 < v < c_5$, $\text{opt}_C(v)$ is equal to $\text{grd}_{C_3}(v-c_4)+1$ or $\text{grd}_{C_3}(v)$.
Now we have $\text{grd}_{C_3}(v-c_4)+1 \leq \text{grd}_{C_3}(v)$, which implies $\text{opt}_C(v) =\text{grd}_{C_3}(v-c_4)+1 = \text{grd}_C(v)$, and thus $v \ (c_4 < v < c_5)$ is not a counterexample~to~$C$.
\\

Suppose $c_5 < v < c_6$.
Since $v-c_5 < c_4-c_2$, $\text{grd}_{C}(v) = \text{grd}_{C_3}(v-c_5)+1$ holds.
In addition, we have $v-c_4 < c_4$, and thus $\text{opt}_C(v)$ is equal to $\text{grd}_{C_3}(v-c_5)+1$, $\text{grd}_{C_3}(v-c_4)+1$, or $\text{grd}_{C_3}(v)$.
We prove that $\text{opt}_C(v) = \text{grd}_{C_3}(v-c_5)+1 = \text{grd}_{C}(v)$ for $c_5 < v < c_6$ by showing $\text{grd}_{C_3}(v-c_5)+1 \leq \text{grd}_{C_3}(v)$ and $\text{grd}_{C_3}(v-c_5)+1 \leq \text{grd}_{C_3}(v-c_4)+1$.

Assume that $v\ (c_5 < v < c_6)$ is the minimum counterexample to $C$.
From Lemma~\ref{Lem:coin1}, we can set $v=pc_2+qc_4$ where $p \in \mathbb{Z}_{> 0}$ and $q \in \{0,1\}$.
In the following, we first consider the case of $q=0$ and then that of $q=1$.

Suppose $q=0$, that is, $v = pc_2$.
In addition, let $c_4 = sc_2 - t$ where $s \in \mathbb{Z}_{>0}$ and $0 \leq t < c_2$.
Since $v-c_4 > c_5-c_4 = c_2$, we have $p \geq s+1$.
The values of $\text{grd}_{C_3}(v-c_4)$ and $\text{grd}_{C_3}(v-c_5)$ depend on the parity of $p-s\/$:
\begin{align}
\text{grd}_{C_3}(v-c_4) &= \text{grd}_{C_3}((p-s)c_2+t) \notag \\
&=
\begin{cases}
\ (p-s)/2+t      & (p-s \text{ is even}) \\
\ (p-s-1)/2+1+t  & (p-s \text{ is odd})
\end{cases}, \label{EQ5} \\
\text{grd}_{C_3}(v-c_5) &=  \text{grd}_{C_3}((p-s-1)c_2+t) \notag \\
&=
\begin{cases}
\ (p-s-2)/2+1+t      & (p-s \text{ is even}) \\
\ (p-s-1)/2+t    & (p-s \text{ is odd})
\end{cases}. \label{EQ6} 
\end{align}
From \eqref{EQ5} and \eqref{EQ6}, $\text{grd}_{C_3}(v-c_5) \leq \text{grd}_{C_3}(v-c_4)$ holds.

As for $c_4 < v < c_5$, the values of $\text{grd}_{C_3}(v)$, $\ell$, and $\text{grd}_{C}(\ell c_3)$ are given as follows:
\begin{align}
\text{grd}_{C_3}(v) &= 
\begin{cases}
\   p/2       & (p \text{ is even}) \\
\  (p-1)/2+1  &  (p \text{ is odd})
\end{cases}, \label{EQ7}
\\
\ell &=
\begin{cases}
\ s/2+1      & (s \text{ is even}) \\
\ (s-1)/2+1  & (s \text{ is odd})
\end{cases}, \label{EQ8}
\\
\text{grd}_{C}(\ell c_3) &=
\begin{cases}
\ t+2      & (s \text{ is even}) \\
\ t+1      & (s \text{ is odd})
\end{cases}. \label{EQ9}
\end{align}
Rearranging \eqref{EQ6}--\eqref{EQ9}, we have
\begin{align}
\text{grd}_{C_3}(v) - (\text{grd}_{C_3}(v-c_5)+1) 
&= \begin{cases}
\ p/2-((p-s)/2+t+1)               & (p \text{ and } s \text{ are even}) \\
\ p/2-((p-s-1)/2+t+1)             & (p \text{ is even, } s \text{ is odd})  \\
\ (p+1)/2-((p-s-1)/2+t+1)         & (p \text{ is odd, } s \text{ is even})  \\
\ (p+1)/2 -((p-s)/2 + t+ 1)       & (p \text{ and } s \text{ are odd})  
\end{cases} \notag\\
&= \begin{cases}
\ s/2-(t+1)              & (p \text{ and } s \text{ are even}) \\
\ (s+1)/2-(t+1)          & (p \text{ is even, } s \text{ is odd})  \\
\ (s+2)/2-(t+1)          & (p \text{ is odd, } s \text{ is even})  \\
\ (s+1)/2 -(t+1)         & (p \text{ and } s \text{ are odd})  
\end{cases} \notag\\
&= \begin{cases}
\ \ell - \text{grd}_C(\ell c_3)+1        & (p \text{ is odd, } s \text{ is even})  \\
\ \ell - \text{grd}_C(\ell c_3)          & (\text{otherwise})  
\end{cases}. \notag
\end{align}
Since $\text{grd}_C(\ell c_3) \leq \ell$, $\text{grd}_{C_3}(v-c_5) +1\leq \text{grd}_{C_3}(v)$ holds for $c_5 < v < c_6$ and $v=pc_2$.

Suppose $q=1$, that is, $v = pc_2 + c_4$.
Since $v > c_5=c_2+c_4$, we have $p \geq 2$.
The values of $\text{grd}_{C_3}(v-c_4)$ and $\text{grd}_{C_3}(v-c_5)$ depend on the parity of $p\/$:
\begin{align}
\text{grd}_{C_3}(v-c_4) &= \text{grd}_{C_3}((pc_2+c_4)-c_4) \notag \\
                        &= \text{grd}_{C_3}(pc_2)  \notag \\
                        &= 
\begin{cases}
\  p/2      & (p \text{ is even}) \\
\ (p-1)/2+1   & (p \text{ is odd})
\end{cases}, \label{EQ10} \\
\text{grd}_{C_3}(v-c_5) &=  \text{grd}_{C_3}((pc_2+c_4)-(c_2+c_4)) \notag \\
                        &= \text{grd}_{C_3}((p-1)c_2) \notag \\
                        &=
\begin{cases}
\ p/2      & (p \text{ is even}) \\
\ (p-1)/2    & (p \text{ is odd})
\end{cases}. \label{EQ11} 
\end{align}
From \eqref{EQ10} and \eqref{EQ11}, we obtain $\text{grd}_{C_3}(v-c_5) +1 \leq \text{grd}_{C_3}(v-c_4) +1$.

Let $c_4 = s'c_2 + t'$ where $s' \in \mathbb{Z}_{>0}$ and $0 \leq t' < c_2$.
Then, from the assumption $c_4 \geq 3c_2-1$, we have $s' \geq 2$.
The value of $\text{grd}_{C_3}(v)$ is given as follows:
\begin{align}
\text{grd}_{C_3}(v) &= \text{grd}_{C_3}(pc_2+c_4) \notag \\
                    &= \text{grd}_{C_3}((p+s')c_2+t') \notag \\
                    &= 
\begin{cases}
\  (p+s')/2 +t'     & (p+s' \text{ is even}) \\
\ (p+s'-1)/2+1+t'   & (p+s' \text{ is odd})
\end{cases}.  \notag
\end{align}
Thus, $\text{grd}_{C_3}(v)$ is at least $(p+s')/2 +t'$ where $s'\geq 2$ and $t' \geq 0$.
From~\eqref{EQ11}, the value of $\text{grd}_{C_3}(v-c_5)+1$ is at most $p/2+1$.
Therefore $\text{grd}_{C_3}(v-c_5)+1 \leq \text{grd}_{C_3}(v)$ holds. \\

From the above discussion, we conclude that $C$ is tight, which directly implies that $C'=(1,c_2,c_3,c_4,c_5)$ is also tight.
In addition, $C_3$ is canonical and $C'$ is noncanonical.
From Theorem~\ref{Thm:tight}, if $C$ is noncanonical, then there exist $i$ and $j$ such that $1 < i \leq j \leq 5$, $c_i + c_j > c_{6} = 2c_4$, and $c_i + c_j$ is a counterexample to~$C$.
Only $(i,j) =(4,5)$ and $(5,5)$ can be such pairs, but $\text{opt}_C(c_4+c_5) > 1$ and $\text{opt}_C(c_5+c_5) > 1$ because $c_4 + c_5 > c_6$ and $c_5 + c_5 > c_6$, respectively.
On the other hand, $\text{grd}_C(c_4+c_5) = 2$ because $c_4 + c_5 = c_6 + c_2$, and $\text{grd}_C(c_5+c_5) = 2$ because $c_5 + c_5 = c_6 + c_3$.
Thus $c_4+c_5$ and $c_5+c_5$ are not counterexamples to $C$, and $C$ is canonical.
\end{proof}





\end{document}